\documentclass[reqno,12pt,letterpaper]{amsart}
\usepackage{amsmath,amssymb,amsthm,graphicx,mathrsfs,url}
\usepackage[usenames,dvipsnames]{color}
\usepackage[colorlinks=true,linkcolor=Red,citecolor=Green]{hyperref}
\usepackage{amsxtra}
\usepackage{epstopdf}
\usepackage{graphicx}
\newcommand\smallO{
  \mathchoice
    {{\scriptstyle\mathcal{O}}}% \displaystyle
    {{\scriptstyle\mathcal{O}}}% \textstyle
    {{\scriptscriptstyle\mathcal{O}}}% \scriptstyle
    {\scalebox{.7}{$\scriptscriptstyle\mathcal{O}$}}%\scriptscriptstyle
  }

\def\?[#1]{\textbf{[#1]}\marginpar{\Large{\textbf{??}}}}

\def\smallsection#1{\smallskip\noindent\textbf{#1}.}
\let\epsilon=\varepsilon % sorry Knuth

\newcommand{\RR}{{\mathbb R}}
\newcommand{\ZZ}{{\mathbb Z}}

\newcommand{\CC}{{\mathbb C}}

\newtheorem{theo}{Theorem}
\newtheorem{prop}{Proposition}[section]

\newtheorem{lemm}[prop]{Lemma}

\newtheorem{rem}{Remark}

\numberwithin{equation}{section}

\DeclareMathOperator{\Spec}{Spec}

\let\Im=\Imag

\let\Re=\Real

\title[Spectral gap for networks of oscillators]{The optimal spectral gap for regular and disordered harmonic networks of oscillators}
\author{Simon Becker}
\email{simon.becker@damtp.cam.ac.uk}
\address{DAMTP, University of Cambridge, Wilberforce Rd, Cambridge CB3 0WA, UK}
\author{Angeliki Menegaki}
\email{angeliki.menegaki@dpmms.cam.ac.uk}
\address{DPMMS, University of Cambridge, Wilberforce Rd, Cambridge CB3 0WA, UK}

\begin{document}
\maketitle
%%%%%%%%%%%%%%%%%%%%%%%%%%%%%%%%%%%%%%%%%%%%%%%%%%%%%%%%%%%%%%%%%%%%%%%%%%%%%%%%
%                                 INTRODUCTION                                 %
%%%%%%%%%%%%%%%%%%%%%%%%%%%%%%%%%%%%%%%%%%%%%%%%%%%%%%%%%%%%%%%%%%%%%%%%%%%%%%%%
%\addtocounter{section}{1}
\begin{abstract}
We consider one-dimensional chains and multi-dimensional networks of harmonic oscillators coupled to two Langevin heat reservoirs at different temperatures. Each particle interacts with its nearest neighbors by harmonic potentials and all individual particles are confined by harmonic potentials, too. In this article, we provide, for the first time, the sharp $N$ dependence of the spectral gap of the associated generator under various physical assumptions and for different spatial dimensions. Our method of proof relies on a new approach to analyze a non self-adjoint eigenvalue problem involving low-rank non-hermitian perturbations of auxiliary discrete Schr\"odinger operators.
\end{abstract}
\tableofcontents

\section{Introduction}
We analyze the dependence of the $L^2$-spectral gap of the full Fokker-Planck operator for a classical heat conduction model from non-equilibrium statistical mechanics by using novel ideas from scattering \cite{SZ} and random matrix theory \cite{FYODOROV199746} to reduce it to a non self-adjoint eigenvalue problem involving discrete Schr\"odinger operators. Even though non self-adjoint eigenvalue problems are often difficult to treat using perturbative methods, we show that the low-rank nature of the non self-adjoint perturbation allows precise estimates on the behaviour of the spectral gap. 

 \subsection{Description of the model} 
In this article we study the so-called \emph{chain of oscillators},  which is a multi-dimensional model\footnote{although in higher dimensions the model is no longer a \emph{chain} of oscillators, but rather a \emph{network}, we shall still use the expression \emph{chain of oscillators} to refer to the model as it was first considered in one dimension and the name \emph{chain of oscillators} has been used pars pro toto.} describing heat transport through a configuration of $N^d$ interacting particles, where $d$ is the spatial dimension. 

\medskip

We assume particles situated on a $d$-dimensional square lattice $[N]^d$, where $[N]:=\{1,..,N\}$, with \emph{quadratic} nearest neighbor interaction and pinning potentials confining the particles of mass $m_i$ to a lattice structure. Let $\textbf{m}_{[N]^d}:=\operatorname{diag}(m_1,...,m_{N^d})$ be the mass matrix, containing the masses $m_i$ of particles $i \in [N^d]$ on the diagonal, and let $q_i$ be the displacement of each particle with respect to their equilibrium position and $p_i$ its momentum.
The energy of the oscillator chain is described by a Hamilton function $H:T^*\mathbb R^{dN^d} \rightarrow \RR$
\begin{equation}
\begin{split}
H(\textbf{q},\textbf{p}) &= \frac{\langle \textbf{p},\textbf{m}_{[N]^d}^{-1}\textbf{p}\rangle}{2}+ V_{\eta,\zeta}(\textbf{q}) \text{ where } \\
V_{{\bf \eta,\zeta}}(\textbf{q}) &= \sum_{i=1}^{N^d} \eta_i \vert q_i \vert^2 + \sum_{i \sim j} \xi_{i,j} \vert q_i-q_{j} \vert^2
\end{split}
\end{equation} 
where $\sim$ indicates nearest neighbors on the $[N]^d \subset \mathbb{Z}^d$ lattice and $\eta_i, \xi_{i,j}>0$. The above form of the potential describes particles that are fixed  by a quadratic \textit{pinning} potential $U_{\operatorname{pin},i}(q) = \eta_i \vert q_i \vert^2 $ and interact through a quadratic \textit{interaction} potential $U_{\operatorname{int},i  \sim j}(q) =  \xi_{i,j} \vert q_i-q_j \vert^2$ for $i,j$ such that $\Vert i-j \Vert_{\infty}=1$. \\
The dynamics of this model is such that (some) particles at the boundary on $ \{1,..,N\}^d$ are coupled to heat baths at (possibly) different temperatures $\beta^{-1}$. Moreover, some particles $i \in I \subset \{1,..,N\}^d$ are subject to friction and we denote by $\gamma_i> 0$ the friction strength at the $i$-th particle. 
\medskip

The time evolution is then for particles $i \in \{1,..,N\}^d $ described by a coupled system of SDEs:

\begin{equation}
\begin{split}
\label{eq:SDE}
dq_i(t)&= \partial_{p_i} H \ dt \text{ and }\\
dp_i(t)&= \left(-\partial_{q_i} H- \gamma_i p_i \delta_{i \in I} \right) \ dt+ \sum_{i \in I}  \sqrt{2m_i \gamma_i \beta_i^{-1}}\ dW_i 
\end{split}
\end{equation}
where $\beta_i$ is the inverse temperature at the boundary of the network of oscillators, $W_i$ with $i \in I$ are iid Wiener processes, $\gamma_i> 0$ a friction parameter, and $I \subset  \{1,..,N\}^d$ the set of the particles subject to friction.

For the analysis of one-dimensional chains, we mainly consider friction at both terminal ends, \textit{i.e.}  $I=\{1,N\}$, in which case $\beta_1$ and $\beta_N$ correspond to actual physical inverse temperatures. Our analysis also allows us to study a chain with zero friction at a single end of the chain, this is a scenario that has been considered by Hairer \cite{Hair09}. In this case, the frictionless end is interpreted to be in contact with an environment at infinite temperature. In this case the inverse temperature at the frictionless end no longer corresponds to a physical temperature.

The solution to the above system of SDEs \eqref{eq:SDE} forms a Markov process, and can thus be equivalently described by a strongly continuous semigroup $P_t f(z) :=   \mathbb{E}_z \big( f( p_t,q_t) \big)$ where $ (p_t,q_t) \in \mathbb{R}^{2N^d}$ solve the system of SDEs \eqref{eq:SDE}. Its generator is given by
\begin{equation}
\label{eq:L}
 \mathcal{L} f(z)  = - z M_{[N]^d} \cdot  \nabla_z f(z) + \nabla_p \cdot  \Gamma \textbf{m}_{[N]^d}  \Theta \nabla_p  f(z)
 \end{equation}
where $M_{[N]^d}  \in \mathbb C^{2N^d \times 2N^d}$ and $\Gamma \in \mathbb R^{N^d \times N^d}$ are matrices of the form
\begin{equation*}
\begin{split}
M_{[N]^d}  &:= \left(\begin{matrix} \Gamma & -\textbf{m}_{[N]^d}^{-1} \\ B_{[N]^d}  & 0 \end{matrix}\right) \text{ and }
\Gamma = \text{diag}(\gamma_1 \delta_{1 \in I}, \dots, \gamma_{N^d} \delta_{N^d \in I}).
\end{split}
\end{equation*}
Here, $\Gamma$ is the friction matrix with $\gamma_i>0$ and $\Theta = \text{diag}(\beta_1^{-1}\delta_{1 \in I},\dots,\beta_{N^d}^{-1}\delta_{N^d \in I}) $ contains the temperatures of the bath.

Defining for $i,j \in [N^d]$ self-adjoint operators $\langle u, L_{i,j}u \rangle_{\ell^2(\mathbb C^{N^d})} := (u(i)-u(j))^2$ that decompose the negative weighted Neumann Laplacian on $\mathbb C^{N^d}$ as $$-\Delta_{N^d} = \sum_{i\sim j} \xi_{i, j} L_{i, j},$$
we can write the operator $B_{[N]^d}$ appearing in $M_{[N]^d}$ as a Schr\"odinger operator 
\begin{equation}
\label{eq:Schroe}
B_{[N]^d} =-\Delta_{[N]^d} + \sum_{i=1}^{N^d} \eta_i \delta_i
\end{equation}
where $(\delta_i(u))(j)=u(i)\delta_{ij}.$
The operator $B_N$ reduces in one dimension to the Jacobi (tridiagonal) matrix 
\[(B_{N} f)_n = -\xi_n f_{n+1}- \xi_{n-1}f_{n-1} +(\eta_n+(2-\delta_{n \in \{1,N\}})\xi_n)f_n\] 
with the convention that $f_0=f_{N+1}=0$.

\subsection{State of the art and motivation}  The (multi-dimensional) chain of oscillators is a non-equilibrium statistical mechanics model initially introduced to study heat transport in media. It was first introduced for the rigorous derivation of Fourier's law, or a rigorous proof of its breakdown: this is well described in several overview articles on the subject: \cite{BLR00}, \cite{Lep16, Dhar08} and \cite{BF19}. The linear (harmonic) case was the first to be studied in \cite{RLL67}, where the non equilibrium steady state (NESS) was explicitly constructed and the behavior of the heat flux analyzed as well, leading (as expected) to the breakdown of Fourier's law.  For results regarding on chains with anharmonic potentials, we refer the reader to \cite{EPR99a, EPR99b, EH00} where existence and uniqueness of stationary states was studied and to \cite{RBT02, Car07} where exponential convergence towards the NESS has been proved. 
Regarding the existence, uniqueness of a NESS and exponential convergence towards it in more complicated anharmonic $d$-dimensional networks of oscillators (not only for square lattices) see \cite{CEHRB18}. In \cite{RAQ18, Me20} bounded perturbations of the harmonic chain are discussed.  
Note also that short chains of rotors with Langevin thermostats have been studied in \cite{CP17, CEP15}. 
 In the articles \cite{HM09, Hair09} some negative results are presented, \textit{i.e.} lack of spectral gap, in cases where the pinning potential is stronger than the coupling one.

 The main motivation of this article is to find the exact scaling of the spectral gap of the associated generator of the dynamics as defined above, in terms of the number of the particles. Quantitative results in this sense are missing from the literature and even in the simplest cases for the chain of oscillators, \textit{i.e.} the linear (harmonic) chains, the dependence on the dimension of the spectral gap. Attempts have been made through hypocoercive techniques to get $N$-dependent estimates under certain assumptions on the potentials: see the discussion in \cite[Section 9.2]{Villani09} where this question was first raised. The techniques discussed in Villani's monograph however only yield rather far from optimal estimates on the spectral gap in terms of the system size. To the authors' knowledge, the only relevant result so far that gives a polynomial lower bound on the spectral gap for the same model (homogeneous with a weak $N$-dependent anharmonicity) is \cite{Me20}. Hypocoercive techniques used in that article provide a polynomial lower bound on the spectral gap and upper bounds on the prefactors in front of the exponential that determine the exponential rate of the convergence. 

Here we give the sharp upper and lower bounds on the scaling of the spectral gap. In this article we not only cover homogeneous networks of oscillators, but also randomly perturbed pinning potentials or pinning potentials perturbed by single impurities. In addition, our techniques also apply to other scenarios apart from the classical one-dimensional model, in particular it gives scalings for $d$-dimensional square network cases. These results seem to be the first of this kind.

\medskip 

\smallsection{Microscopic properties and heat transport} Before stating our main results, we want to mention results on the macroscopic heat transport of the chain of oscillators, \textit{e.g.} heat conductivity, and how such properties are determined from microscopic properties of the system. In particular, we would like to highlight which microscopic properties affect the heat transport and which determine the asymptotic behaviour of the spectral gap. 

\medskip

It has been suggested by \cite{CL71} that, for an infinite chain the absolutely continuous part of the spectrum of the Schr\"odinger operator \eqref{eq:Schroe}, \textit{i.e.} the \emph{metallic part} of the spectrum, leads to infinite conductivity. In the specific example of the homogeneous chain, where there is only absolutely continuous spectrum in the limit, it is well-known that the conductivity is infinite (Fourier's law doesn't hold) \cite{RLL67}. 
Note also that the behavior of the flux does not depend on the dimensionality of the system, see \cite{Hellemann} for $2$ dimensions.
 However, in disordered harmonic chains (DHC) with random masses, where all eigenstates of the discrete Schr\"odinger operator are localized, the heat flux vanishes as $N \rightarrow \infty$ almost surely, see \cite{CL71, RG71, CL74}. In terms of the conductivity that is $\frac{\kappa(N)}{N} \to 0$ as $N$ goes to infinity.  

\medskip 

%The validity of Fourier's law for $d$-dimensional chains of harmonic (and anharmonic) chains perturbed by an energy (and possibly also momentum)-conserving noise, has been studied in \cite{Bern09}. While our analysis of the spectral gap can be performed similarly for other types of boundary conditions, e.g. fixed or periodic, with the same quantitative results, it has been observed in \cite{Bern09} that boundary conditions can change the behaviour of the conductivity\textcolor{red}{, when working with perturbed chains by noise conserving energy.}.

 First studies of the behaviour of the heat currents in a one-dimensional DHC were done in \cite{CL71, RG71}. In particular, in \cite{RG71} the heat baths are semi infinite harmonic chains distributed with respect to Gibbs measures at temperatures $T_L, T_R$ (free boundaries). In this case, $\mathbb{E}(J_N) \gtrsim N^{-1/2}$, where $\mathbb{E}(\cdot)$ denotes the expectation over the masses. That $\mathbb{E}(J_N) \sim N^{-1/2}$ was proved a bit later in \cite{Ver79}, showing that Fourier's law does not hold in this model of DHC.  Results regarding heat baths coupled at both ends with Ornstein-Uhlenbeck terms with  fixed boundaries, \textit{i.e.} $q_0=q_{N+1}=0$, was first done in \cite{CL71}. A rigorous proof of $\mathbb{E}(J_N) \sim (\Delta T)  N^{-3/2}$ was given in \cite{AH10}. The limiting behaviour of the heat flux in both of these models is also discussed in \cite{Dhar01}. Localization effects of the discrete Schr\"odinger operator enter also in the study of mean-field limits for the harmonic chain \cite{BHO19}.

Our new approach shows that the spectral gap of the generator to \eqref{eq:SDE} is determined by the decay rate of eigenstates of the discrete Schr\"odinger operator \eqref{eq:Schroe}
\[ (B_{[N]^d} f)_i=(-\Delta_{[N]^d} f)_i + \eta_i f_i, \text{ where } f = (f_i)_{i \in [N^d] },\] 
defined in terms of masses, the potential coupling strengths, under a constraint on the level-spacing between its eigenvalues. 
%Here, $ -\Delta_{N^d}$ is a discrete Laplacian on $[N]^d$. 
In particular, our results indicate that the presence of exponentially localized eigenstates in the discrete Schr\"odiger operator, \textit{i.e.} the \emph{insulating part} of the spectrum, causes an exponentially fast closing of the spectral gap. In contrast to this, if the discrete Schr\"odinger operator possesses only extended states, the spectral gap again decays to $0$ as $N$ tends to infinity but this time only at a polynomial rate. Both results only hold under a pressure condition on the eigenvalues. 

The above results show that single impurities which correspond to rank one perturbations in the discrete Schr\"odinger operator should not affect the heat conductivity but do affect the spectral gap. Put differently, heat transport is an effect that is governed by all the modes of the system whereas the spectral gap is -in general- only determined by a single extremizing mode of the Schr\"odinger operator.\\

\subsection{Main results}

We study the spectral gap for three scenarios describing fundamentally different physical settings:
 \begin{itemize}
 \item For a homogeneous model with the same physical parameters for every particle (the associated Schr\"odinger operator possesses only extended states in the limit $N \rightarrow \infty$), Fig. \ref{fig:classic}, 
 \item for a model with a sufficiently strong impurity in the pinning potential of a single particle (the Schr\"odinger operator possesses both extended and exponentially localized states in the limit $N \rightarrow \infty$), Fig. \ref{fig:flying}, and 
 \item for a model with disordered pinning potential (the Schr\"odinger operator has only exponentially localized eigenstates in the limit $N \rightarrow \infty$ for $d=1$, this is also conjectured to be true for $d=2$, and is conjectured to have both exponentially localized and extended states in dimensions $d \ge 3$), Fig. \ref{fig:disorder}. 
 \end{itemize}

 Our main results on the $N$-dependence of the spectral gap of the $d$-dimensional harmonic chain are summarized in the following theorem:

\begin{figure}
\centering
\begin{minipage}{0.80\textwidth}
\includegraphics[width=10cm]{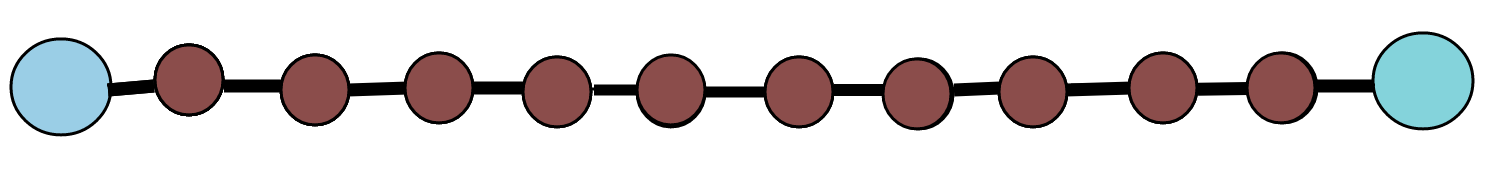}
\caption{Homogeneous chain: Spectral gap $\sim N^{-3}$.}
\label{fig:classic}
\end{minipage}
\begin{minipage}{0.80 \textwidth}
\includegraphics[width=10cm]{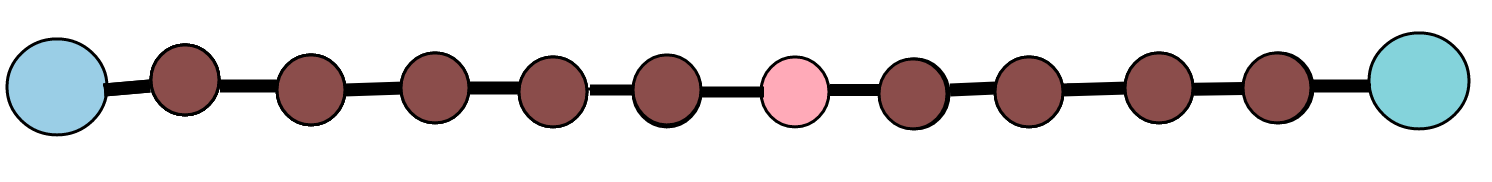} 
\caption{Chain with impurity: Spectral gap $\sim e^{-cN}$.}
\label{fig:flying}
\end{minipage}
\begin{minipage}{0.80 \textwidth}
\includegraphics[width=10cm]{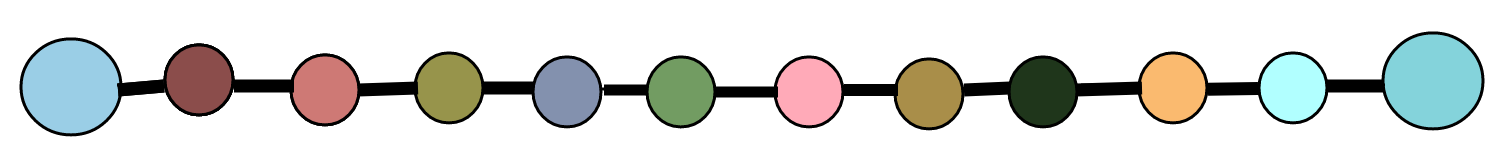} 
\caption{Disordered chain: Spectral gap $\sim e^{-cN}$.}
\label{fig:disorder}
\end{minipage}
 \label{pic: line}
 \caption{The one-dimensional chain of oscillators connected to heat baths (big discs) and with various pinning potentials (differently coloured discs indicate different pinning strengths).}
\end{figure}
\begin{theo}
\label{theo:main}
Let the positive masses and interaction strengths of all oscillators coincide, $N$ be the number of oscillators and $d$ the dimension of the network. Then if the sum of all friction parameters for all oscillators is uniformly bounded, the spectral gap of the chain of oscillators closes always as a function of $N$. 
%at least with rate $\mathcal O(1/N).$  
In particular, we have the following cases
\begin{itemize}
\item \emph{(Homogeneous chain):} If the pinning strength is the same for all oscillators, 
\begin{enumerate} 
\item    when the two particles located at the corners $(1,\dots,1),(N,\dots,N)$ are exposed to non-zero friction and diffusion,  the spectral gap of the generator satisfies $$N^{-3d}\lesssim \lambda_S \lesssim N^{-3d}.$$ In particular for the one-dimensional chain we have $N^{-3} \lesssim \lambda_S \lesssim N^{-3}$. 
\item    when the friction and diffusion act on the two particles located at the center of the two edges of the network $(1,  \lceil N/2 \rceil, \dots,  \lceil N/2 \rceil), (N,  \lceil N/2 \rceil, \dots,  \lceil N/2 \rceil)$,  the spectral gap of the generator satisfies $$N^{-3-(d-1)}\lesssim \lambda_S \lesssim N^{-3-(d-1)}.$$
\item    when $d=2$ and the particles exposed to friction are located on the whole two opposite edges, the spectral gap then satisfies $\lambda_S \lesssim N^{-5/2}$.
\end{enumerate} 
\item \emph{(Chain with impurity):} Let $N$ be even. We assume that all masses and interaction parameters are positive and coincide and the friction parameters $\gamma_i$ of the particles at the boundary  
\[ \partial [N]^d:=\{i \in [ N]^d; \exists i_n: i_n=1 \text{ or } i_n=N\} \text{ of }[N]^d:=\{1,..,N\}^d\] 
satisfy $\sup_{i \in \partial [ N]^d} \gamma_i \in (0,c)$where $c$ is independent of $N$. Then, if the pinning strength  $$\eta_{c_d(N)}$$ at the center point $c_d(N)=(N/2,..,N/2)$ of the network is sufficiently small compared to the pinning strength of all other oscillators, then the spectral gap $\lambda_S$ of the generator closes exponentially fast in the number of oscillators, for all $d \ge1$. 
\item \emph{(Disordered chain):} We assume that all masses and interaction parameters are positive and coincide and the friction parameters $\gamma_i$ of the particles at the boundary  
\[\partial [\pm N]^d:=\{i \in [\pm N]^d; \Vert N \Vert_{\infty}=N\} \text{ of the network }[\pm N]^d:=\{-N,...,N\}^d\] 
satisfy $\sup_{i \in \partial [\pm N]^d} \gamma_i \in (0,c)$ where $c$ is independent of $N$. Then, if the pinning strengths are iid random variables according to some compactly supported density $\rho \in C_c(0,\infty)$, the spectral gap $\lambda_S$ of the generator closes exponentially fast in the number of oscillators, for all $d \ge 1$ for all but finitely many $N$. 
\end{itemize}
\end{theo} 
Our findings in Theorem \ref{theo:main} are illustrated in one spatial dimension in Figure \ref{fig:clement}.
\begin{figure}
\includegraphics[width=10cm]{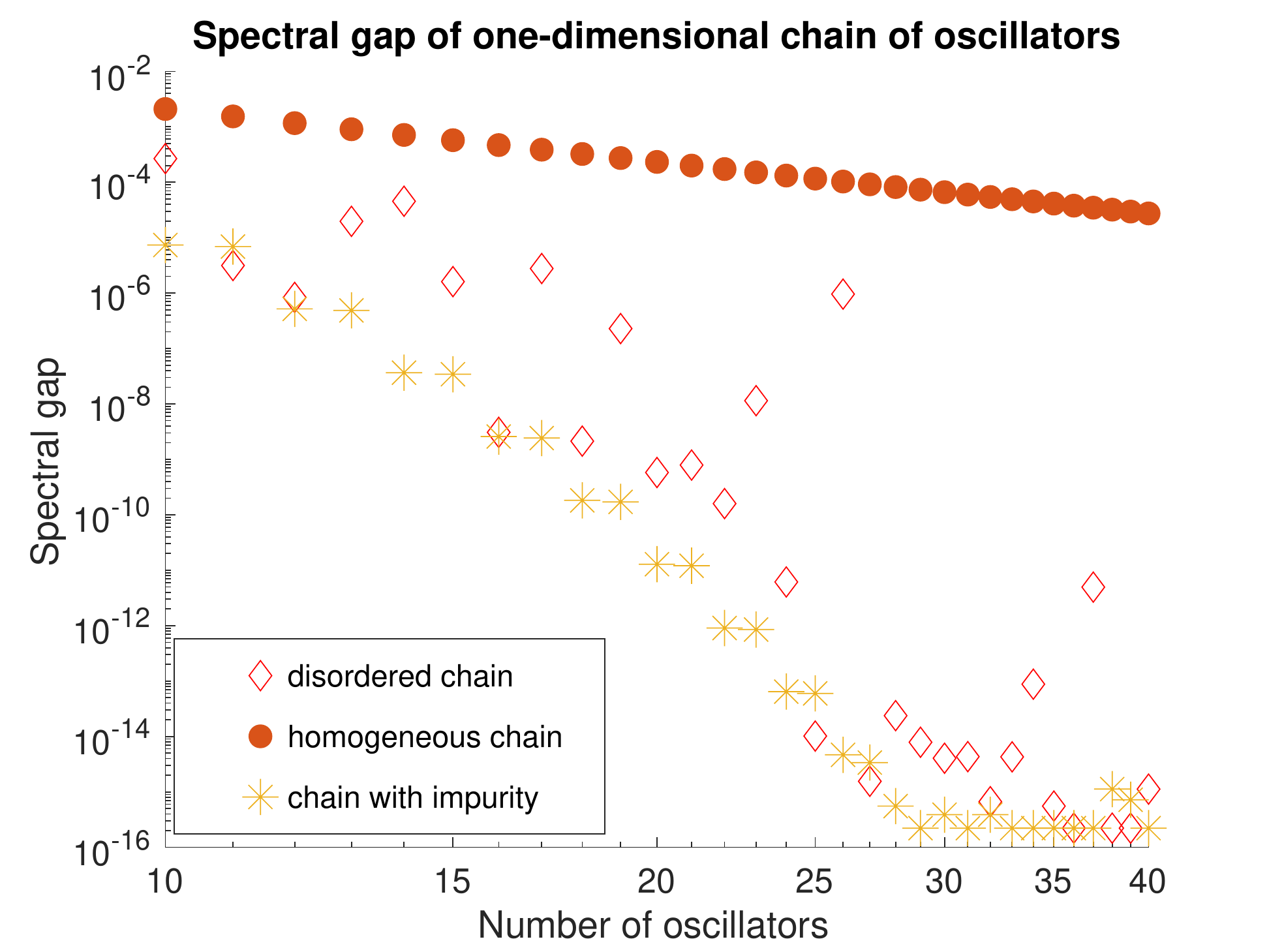}
\caption{Log-log plot of the spectral gap for the one-dimensional chain of oscillators for all three cases considered in Theorem \ref{theo:main}. The impurity is modeled by choosing a pinning strength $\eta_i=10$ for all oscillators $i$ apart from the one in the center for which we choose $\eta_{N/2}=5.$ The disorder potential is of the form $V_{\omega}(n)=1+X_n$ where $X_n \sim U(0,1)$ are i.i.d. uniformly distributed.} 
\label{fig:clement}
\end{figure}

\smallsection{Open questions}
\begin{itemize}
\item While we fully settle the scaling of the spectral for one-dimensional oscillator chains, the scaling of the spectral gap for many physically relevant configurations in higher dimensions remains open. Although our method of proof still applies to such configurations as well, the necessary estimates seem to become rather intricate, cf. the discussion below Proposition \ref{prop:hom2}.
\item It would be interesting to study the behavior of the spectral gap in terms of the dimension of the system in the oscillator chains for more general classes of pinning and interaction potentials, \textit{i.e.} for nonlinear chains. While this analysis cannot be reduced to a Schr\"odinger operator in that case, we still believe the connection between decay properties of (generalized) eigenstates of the symmetric part of the operator and the scaling of the spectral gap to persist. 
\item Moreover, apart from considering different kinds of potentials, one can study different kind of noises as well, \cite{RAQ18, NR19}, where quantitative rates of convergence are not available, so far.
\item It would also be interesting to extend our analysis to more complicated geometries such as different lattice structures.
\end{itemize}

%Chains with potentials that behave like polynomials at infinity have been studied before in \cite{EPR99b, EPR99a, RBT02, Car07} and also recently in \cite{CEHRB18} treating more complicated networks of oscillators. There, questions regarding the existence, uniqueness and exponential convergence in time towards the non-equilibrium steady state (NESS) have been answered. 
%In \cite{RAQ18, Me20} bounded perturbations of the harmonic chain are discussed.

%\medskip
%
% In articles \cite{HM09, Hair09} some negative results are presented, \textit{i.e.} lack of spectral gap, in cases where the pinning potential is stronger than the coupling one.  

\smallsection{Notation}
We write $f(z) = \mathcal O(g(z)) $ to indicate that there is $C>0$ such that $\left\lvert f(z) \right\rvert \le C \left\lvert g(z) \right\rvert$ and $f(z)= \smallO(g(z))$ for $z \rightarrow z_0$ if there is for any $\varepsilon>0$ a neighborhood $U_{\varepsilon}$ of $z_0$ such that $\left\lvert f(z)\right\rvert \le \varepsilon \left\lvert g(z) \right\rvert.$ 
Finally, we introduce the notation $[N]:=\left\{1,...,N \right\}$ and 
\[\partial [N]^d:=\{ i=(i_1,..,i_d) \in [N]^d; \Vert i \Vert_{\infty}=N \text{ or } \min_{n \in [d]} i_n =1\}.\] 
The eigenvalues of a self-adjoint matrix $A$ shall be denoted by $\lambda_1(A)\le ... \le \lambda_N(A)$. We also employ the Kronecker delta where $\delta_{n \in I}=1$ if $n \in I$ and zero otherwise.

\section{Mathematical preliminaries}
\label{sec:COO}

For our purposes, it is sometimes favorable to consider also another form, which we obtain upon performing the following change of variables $$ \tilde{p}= p\ \sqrt{\textbf{m}_{[N]^d}^{-1}},\  \tilde{q} := q \sqrt{B_{[N]^d}} .$$ This is an isomorphic change of variables if and only if all masses and interaction strength are strictly positive. In the new coordinates, the generator becomes
\begin{equation}
\begin{split}
 \label{eq:generator} 
\mathcal{L}&=\tilde{p}\textbf{m}_{[N]^d}^{-1/2} B_{[N]^d}^{1/2} \nabla_{\tilde{q}} - \tilde{q} B_{[N]^d}^{1/2}\textbf{m}_{[N]^d}^{-1/2} \nabla_{\tilde{p}} -\tilde{p} \Gamma \cdot \nabla_{\tilde{p}} + \nabla_{\tilde{p}} \cdot \Gamma \Theta \nabla_{\tilde{p}}\\ 
&= -\tilde{z} \Omega_{[N]^d} \cdot \nabla_{\tilde{z}}  +\nabla_{\tilde{p}} \cdot \Gamma \Theta \nabla_{\tilde{p}} 
\end{split}
\end{equation}
 where 
 $$ \Omega_{[N]^d} := \left(\begin{matrix} \Gamma  & - \textbf{m}_{[N]^d}^{-1/2}B_{[N]^d}^{1/2} \\ B_{[N]^d}^{1/2} \textbf{m}_{[N]^d}^{-1/2} & 0  \end{matrix}\right).$$
\\

The following Proposition identifies the optimal exponential rate of convergence, and thus the spectral gap, to the NESS for Ornstein–Uhlenbeck operators. This result was first proved, to our knowledge by  \cite{MPP02}.  Here we state a version given in \cite{AE14, Mon15}: 

\begin{prop}[Proposition 13 in  \cite{Mon15}, Theorems 4.6 and 6.1 in \cite{AE14}, Theorem 2.16 in \cite{AAS15}] \label{prop about reduction of sg}
Let the generator of an Ornstein-Uhlenbeck process given by  \begin{align} \label{FP operator} Lf(z) = -(Bz) \cdot \nabla_z f(z) + \operatorname{div}(D \nabla_z f)(z),\quad z \in \mathbb{R}^d \end{align} under the assumptions that 
\begin{enumerate}
\item There is no non-trivial subspace of $\text{Ker} D$ that is invariant under $B^T$   
\item All eigenvalues of the matrix $B$ have positive real part ($B$ is positively stable).
\end{enumerate}
Let $\rho := \inf\{ \Re(\lambda): \lambda \in \Spec( B) \}>0$ and let $m$, that possibly depends on $N$,  be the maximal dimension of the Jordan block of $B$ that corresponds to an eigenvalue  $\lambda$ of $B$ such that $\Re(\lambda)= \rho$. \\
 Then there is a unique invariant measure $\mu$ and constant $c>0$ so that, regarding the long time behaviour of the process with generator \eqref{FP operator},  $$  c^{-1} (1+t^{2(m-1)}) e^{-2\rho t} \le \| P_t - \mu \|^2 \le c e^{\rho} (1+t^{2(m-1)})e^{-2 \rho t}$$ where $  \| P_t- \mu \| : = \sup_{\| f\|_{L^2(\mu)=1} } \{ \| (P_t - \mu)f \|_{L^2(\mu) } \}$. 
\end{prop}

\bigskip

Therefore, both the exponential rate given by $\rho$ is and the power $(1+t^{2(m-1)})$ are optimal.  Now if we define for every $\epsilon \in (0, \rho)$,  $$C_{\epsilon,N}:= \sup_{t>0} e^{-2\epsilon t } (1+ t^{2(m-1)})$$ we have $$  (1+t^{2(m-1)}) e^{-2\rho t} \le  C_{\epsilon,N} e^{-2(\rho-\epsilon)t}.$$ 

%and since this holds for every $\epsilon \in (0, \rho)$, the spectral gap of the operator \eqref{FP operator} is given by $\rho>0$. 
Note that since $m$ can depend on $N$, $C_{\epsilon,N}$  depends on $N$, too. The exponential rate and more generally the estimates of the relaxation time, is due to the drift part of the operator, whereas the hypoellipticity condition is used to ensure us for the existence of a unique invariant measure $\mu$ (in \cite[Lemma 3.2]{AE14} it is established that the invariant measure is in general a non-isotropic Gaussian measure. See also \cite{RLL67} where they find an explicit form of this stationary measure having as motivation to study properties of the NESS of the harmonic oscillators chains.)

We finally would like to mention that such a result holds in relative entropy as well \cite{Mon15}. \\

The above Proposition \ref{prop about reduction of sg} applies to the chain of oscillators as well, where $B$ is just $\Omega_{[N]^d}^T$ in \eqref{eq:generator}. Conditions (1) and (2) are satisfied, once we assume there is diffusion and friction, i.e.\@ $I \neq \emptyset$, since this condition is equivalent to the hypoellipticity of $\partial_t -L$  \cite[§1]{H69}. Also $\Omega_{[N]^d}$  satisfies condition (ii), see \cite[Lemma 5.1]{JPS17}.
Since we don't know if our matrix $\Omega_{[N]^d}$ is diagonalizable, $C_{\varepsilon, N}$ here depends possibly on $N$ and when considering the worst case we have a dependence of order $t^{2(N-1)}$ on the right hand side. Then applying Proposition \ref{prop about reduction of sg} in our case we get $$ 2c^{-1}  e^{-2\rho t} \le \| P_t - \mu \|^2 \le c e^{\rho} (1+t^{2(N^d-1)})e^{-2 \rho t} .$$ 

To summarize the discussion of this Section: The spectral gap $\lambda_S$ of the generator of the $N$-particle dynamics \eqref{eq:generator} is precisely given by 
\[\lambda_S:=\inf\{ \Re(\lambda): \lambda \in \Spec( \Omega_{[N]^d}) \} =\inf\{ \Re(\lambda): \lambda \in \Spec(M_{[N]^d}) \}  .\] 
%\end{rem}

We record some simple observations about the behaviour of the spectral gap in the following Proposition: 
\begin{prop}
\label{prop:generalproperties}
For the one-dimensional chain of oscillators the following properties hold:
\begin{enumerate}
%\item The characteristic polynomial of $M_N$ satisfies $\operatorname{det}(M_N-\lambda \operatorname{id}) = \operatorname{det}(\lambda^2-\lambda \Gamma+\textbf{m}^{-1}B_N)=0.$ 
%In particular, the matrix $M_N$ is invertible if and only if $B_N$ is invertible. 
\item The characteristic polynomial of $M_{[N]^d}$ satisfies $\operatorname{det}(M_{[N]^d}-\lambda \operatorname{id}) = \operatorname{det}(\lambda^2-\lambda \Gamma+\textbf{m}_{[N]^d}^{-1}B_{[N]^d})=0.$ In particular, the matrix $M_{[N]^d}$ is invertible if and only if $B_{[N]^d}$ is invertible. 
\item If there is the same non-zero friction at every oscillator, i.e.\@ $\Gamma=\gamma \operatorname{id}$, $I=[N^d]$, and $\inf_{N \in \mathbb N}\inf_{\lambda \in \Spec(\textbf{m}_{[N]^d}^{-1}B_{[N]^d})}\lambda>0$ then the chain of oscillators has a spectral gap that is uniform in the number of oscillators. In particular, if all masses and coupling strength $\eta,\xi$ coincide and are non-zero, then we have $\inf_{N \in \mathbb N}\inf_{\lambda \in \Spec(\textbf{m}_{[N]^d}^{-1}B_{[N]^d})}\lambda>0$.
\item The spectral gap of the generator \eqref{eq:L} closes at least with rate $\mathcal O(N^{-1})$ if the friction parameters at particles on the boundary $\partial [N]^d$ of the particle configuration $[N]^d$ is uniformly bounded, i.e. $\sup_{i \in I \subset \partial [N]^d} \gamma_i \le C$ where $C>0$ is independent of $N$. 
\item Let $1 \in I$ be the left terminal end of a one-dimensional chain with universal (independent of the size of the chain) friction parameter $\gamma_1 > 0$. If all oscillators have the same mass and there are constants $c_1,c_2>0$ such that $c_1<\xi_i,\eta_i < c_2$ for all $N$, then the spectral gap of \eqref{eq:L} does not close faster than $e^{-cN}$ for some $c>0.$
\end{enumerate}
\end{prop}
\begin{proof}
(1): The determinant formula follows from general properties of block matrices. By setting $\lambda=0$ it follows that $M_{[N]^d}$ is invertible if and only if $\textbf{m}_{[N]^d}^{-1}B_{[N]^d}$. 
\medskip

\noindent
(2): If $I=[N^d]$ and $\Gamma=\gamma I$ then $\operatorname{det}(\lambda^2-\gamma \lambda + \textbf{m}_{[N]^d}^{-1}B_{[N]^d})=0$ is equivalent to solving $\lambda^2- \gamma \lambda +\mu=0$ where $\mu \in \Spec(\textbf{m}_{[N]^d}^{-1}B_{[N]^d}).$ Now as the product of two positive definite matrices, $\textbf{m}_{[N]^d}^{-1}B_{[N]^d}$ has again positive eigenvalues. Thus, all solution to this equation have their real part bounded away from zero.

\noindent
(3) is a consequence of the identity $$\sum_{\lambda \in \Spec(M_{[N]^d})} \lambda= \operatorname{tr}(M_{[N]^d})  = \operatorname{tr}(\Gamma).$$
Since we have $2N^d$ (counting multiplicity) positive terms that all satisfy $\Re(\lambda) \ge 0$ where $\lambda \in \Spec(M_{[N]^d})$, and by assumption $\operatorname{tr}(\Gamma) =\mathcal O(N^{d-1}),$ we conclude that $\lambda_S = \mathcal O(N^{-1})$: Indeed we write $$ \sum_{\lambda \in \Spec(M_{[N]^d})} \Re(\lambda) \ge (2N^d) \inf\{ \Re(\lambda): \lambda \in \Spec(M_{[N]^d}) \} = \mathcal O(N^{d-1})$$ and thus $$\lambda_S=  \inf\{ \Re(\lambda): \lambda \in \Spec(M_{[N]^d}) \} = \mathcal{O}(N^{-1}).$$

\noindent
(4): We introduce the transfer matrix \cite[(1.29)]{T}
\begin{equation}
\label{eq:cocycle}
A_i(\lambda) = \left(\begin{matrix} \frac{ \lambda^2-(\xi_{i}+\xi_{i+1}+\eta_{i+1})}{\xi_{i+1}} & -\frac{\xi_{i}}{\xi_{i+1}} \\ 1 & 0 \end{matrix} \right).
\end{equation}
Thanks to the tridiagonal and symmetric form of $B_N$, the transfer matrix \eqref{eq:cocycle} allows us to write the solution to $(B_N+\lambda^2)u=0$ inductively as
\[ \left( \begin{matrix} u_{i+1} \\ u_i \end{matrix} \right) = \prod_{j=i-1}^1 A_j(\lambda) \left(\begin{matrix} u_{2} \\ u_1 \end{matrix} \right).\] This way, $\left\lVert \left( u_{i+1}, u_i  \right)^T \right\rVert \le C^{i-1} \left\lVert \left(  u_{2},u_1 \right)^T \right\rVert $
with boundary conditions 
\begin{equation*}
\begin{split}
u_2 &= \frac{\left(\lambda^2-\lambda+\eta_{1}+\xi_{1}\right)}{\xi_1} u_1 \text{ and }
u_{N-1}= \frac{\left(\lambda^2-\lambda+\eta_{N}+\xi_{N-1}\right)}{\xi_{N-1}} u_N.
\end{split}
\end{equation*}
and where $ C= \sup_{j} \|A_j\|$ . \\ Let $\lambda \in \Spec(M_N)$ with $\Re(\lambda)= \lambda_S,$ then there is $u$ normalized such that $$(\textbf{m}_N^{-1}B_N+\lambda^2)u=\lambda \Gamma u.$$
In particular, this implies by taking the inner-product with $u$ again: $$\langle (\textbf{m}_N^{-1}B_N+\lambda^2)u,u \rangle = \lambda \langle \Gamma u,u \rangle,$$ and since the real and the imaginary parts in both sides should be the same, we write for the imaginary part $\Im(\lambda^2) = \Im(\lambda) \sum_{i \in I} \gamma_i \vert u_i \vert^2.$
Writing now $\lambda=\lambda_S+ i \Im(\lambda)$ yields 
\begin{equation}
\label{eq:sgap}
\lambda_S  = \sum_{i \in I} \gamma_i \vert u_i \vert^2/2 \ge \frac{\gamma_1 |u_1|^2}{2}.
\end{equation}

Since $u$ is normalized this implies, using also \eqref{eq:sgap}, that 
\[1 = \sum_{i=1}^N \vert u_i \vert^2 \le C_1^{2N}
 \vert u_1 \vert^2 \le 2 \frac{ C_1^{2N}}{\gamma_1} \lambda_S\] which implies the claim as we assumed that there is friction at the first particle.
 \end{proof}
 \begin{rem} The artificial case $(2)$, in which we assume friction at every particle, and the result in $(3)$ show that it is the sub-dimensionality of the particles experiencing friction that causes the spectral to close for almost all configurations of the chain of oscillators. 
 \end{rem} 
 
\section{Proofs of the main results}

\subsection{Reduction method from scattering theory}
In a preliminary step, we harness the low-rank character of the perturbation and reduce the study of the spectral gap to an auxiliary problem.

The following Lemma reduces the dimension of the spectral analysis of $\Omega_{[N]^d} \in \mathbb C^{N^d \times N^d}$, which determines the spectral gap of the generator \eqref{eq:generator}, to an equivalent problem for a low-dimensional \emph{Wigner matrix} $\widetilde{R}_{I} \in \mathbb C^{ \vert I \vert \times \vert I \vert}$ and connects the low-dimensional Wigner matrix to the eigenvectors of the off-diagonal blocks of $\Omega_N$. For more background on this method, that originates from scattering theory, we refer to \cite{SZ}. We apply it here to study the spectra of low-rank perturbations, due to friction at the boundary oscillators, of the Hamiltonian system.

\begin{lemm}[Low-rank perturbations]
\label{lemm:SGI}
Let $B$ be a self-adjoint matrix on $\CC^{N^d}$ with eigenvalues $\lambda_j$ and eigenvectors $v_j$ and consider the matrix $\mathcal A =i\Omega= \mathcal A_0 + i \Gamma \oplus 0_{\mathbb C^{N^d \times N^d}}$ where $\mathcal A_0 = \left(\begin{matrix} 0 & -iB \\ iB & 0 \end{matrix} \right).$ We then have that $\lambda \in \Spec(\mathcal A)$ if and only if $i \in \Spec(\widetilde{R}_{I}(\lambda))$ where for $V_j^{\pm}=\tfrac{1}{\sqrt{2}}( v_j,\pm i v_j)^T$ 
\begin{equation}
\label{eq:Wigner}
\widetilde{R}_{I}(\lambda)= \sum_{ j=1}^{N^d} (\lambda-\lambda_j)^{-1}  \sum_{\pm}\sum_{i_1,i_2 \in I} \sqrt{\gamma_{i_1}\gamma_{i_2}}\langle V_j^{\pm}, e_{i_1}^{2N^d} \rangle \langle e_{i_2}^{2N^d},V_j^{\pm}\rangle e_{i_1}^{\vert I\vert} \otimes e_{i_2}^{\vert I\vert} 
\end{equation}
and $e_j^{n}$ is the $j$-th unit vector in $\mathbb C^n.$
\end{lemm}
\begin{proof}
We define matrices $A_{I}= \left\{\sqrt{\gamma_a} e_a^{2N^d}(i) \right\}_{i \in [2N^d],a \in I} \in \mathbb C^{2N^d \times \vert I \vert}$ and then have that the friction matrix is given by $\Gamma  = A_I A_I^*.$

The Wigner $\widetilde R_{I}$-matrix is defined as $$\widetilde R_{I}(\lambda):=A_I^{*} (\lambda-\mathcal A_0)^{-1}A_{I} \in \mathbb C^{\vert I \vert \times \vert I \vert}.$$

We then obtain from properties of the determinant, and Sylvester's determinant identity in particular, that
\begin{align*}
\operatorname{det}\left( \operatorname{id}_{\vert I \vert} -i \widetilde R_{I}(\lambda) \right) &= \operatorname{det}\left( \operatorname{id}_{\vert I \vert} -i A_I^{*} (\lambda-\mathcal A_0)^{-1}A_{I}  \right) =  \operatorname{det}\left( \operatorname{id}_{\vert I \vert} -i  (\lambda-\mathcal A_0)^{-1} \Gamma \right) \\ &= \operatorname{det} \left( (\lambda-\mathcal A_0)^{-1} (\lambda - \mathcal A_0 - i \Gamma)  \right)\\
&=  \operatorname{det}  \left( (\lambda-\mathcal A_0)^{-1} \right) \operatorname{det} \left(\lambda - \mathcal A \right).
\end{align*}
Rearranging this identity shows that
\begin{equation}
\begin{split}
0=\operatorname{det}(\lambda-\mathcal A) = \operatorname{det}(\lambda-\mathcal A_0) \operatorname{det}\left(\operatorname{id}_{\vert I \vert} - i \widetilde R_{I}(\lambda)\right).
\end{split}
\end{equation}
Thus, all eigenvalues $\lambda$ of the high-dimensional matrix $\mathcal A$ coincide with values $\lambda$ for which $i \in \Spec(\widetilde R_{I}(\lambda)).$ 
The eigenvectors of $\mathcal A_0$ are given by $V_j^{\pm}=\tfrac{1}{\sqrt{2}}( v_j,\pm i v_j)^T$ where 
 $v_j$ are eigenvectors of $B$ to eigenvalues $\lambda_j.$ 
 We thus find the following expression for $\widetilde{R}_I$
\[\widetilde{R}_{I}(\lambda)= \sum_{ j=1}^{N^d} (\lambda-\lambda_j)^{-1}  \sum_{\pm}\sum_{i_1,i_2 \in I} \sqrt{\gamma_{i_1}\gamma_{i_2}}\langle V_j^{\pm}, e_{i_1}^{2N^d} \rangle \langle e_{i_2}^{2N^d},V_j^{\pm}\rangle e_{i_1}^{\vert I\vert} \otimes e_{i_2}^{\vert I\vert}.  \]
 \end{proof}

 \subsection{One-dimensional homogeneous chain}  We first study the behaviour of a one-dimensional chain of oscillators that consists of particles with the same physical properties. The limiting discrete Schr\"odinger operator $B_{\infty}$ possesses only absolutely continuous spectrum, by standard properties of the discrete Laplacian, and we find a polynomially fast rate for the closing of the spectral gap:
 \begin{prop}[Homogeneous chain]
 \label{prop:hom}
Let all pinning and interaction parameters $\eta_i>0$, $\xi_i>0$ of the potentials, and masses $m_i$ coincide, respectively and assume that there is at least one particle with non-zero friction and diffusion at one of the terminal ends of the chain. The spectral gap of the harmonic chain of oscillators satisfies $N^{-3}\lesssim \lambda_S \lesssim N^{-3}.$
 \end{prop}
 \begin{proof}
 
The eigenvectors to the root of the discrete Schr\"odinger operator $\sqrt{B_N}$, defined in \eqref{eq:Schroe}, coincide with the eigenvectors to the discrete Laplacian and are just given by
\begin{equation}
\label{eq:eigenvectors}
v_j(i)= \begin{cases}
N^{- \frac{1}{2}} & , \ \mbox{j = 1}\\
\sqrt{\frac{2}{N}} \cos\left(\frac{\pi (j - 1)\left(i - \tfrac{1}{2}\right)}{N}\right) & , \ \mbox{otherwise}
\end{cases}
\end{equation}
with eigenvalues $$ \lambda_j(\sqrt{B_N}) = \sqrt{4 \xi\sin^2\left(\frac{\pi (j - 1)}{2N}\right)+\eta}$$ of $\sqrt{B_N}$. 
We then define $\mu_{j}:= \lambda_j(\sqrt{B_N})-\lambda_N(\sqrt{B_N})$ and observe that by Taylor expansion we have
\begin{equation}
\label{eq:estmate}
\vert \sqrt{4\xi+\eta}-\lambda_{j}(\sqrt{B_N}) \vert \lesssim \vert \lambda_{N}(\sqrt{B_N})-\lambda_{j}(\sqrt{B_N})\vert
\end{equation}
for $j \le N-1,$ such that by using this estimate in the final step
\begin{equation}
\begin{split}
\label{eq:mui}
 \mu_{j}^{-1}  = \frac{\left\vert  \lambda_j(\sqrt{B_N}) + \lambda_N(\sqrt{B_N}) \right\vert }{\left\vert \lambda_j(\sqrt{B_N})^2 -\lambda_N(\sqrt{B_N})^2 \right\vert}&\lesssim \vert \lambda_j(\sqrt{B_N})^2 -\lambda_N(\sqrt{B_N})^2 \vert^{-1} \\
 &\lesssim \vert \lambda_j(\sqrt{B_N})^2 -(4\xi+\eta) \vert^{-1} \\
 \end{split}
 \end{equation}
This yields by combining \eqref{eq:mui} with the explicit expression of the eigenvalues \eqref{eq:eigenvectors}
\begin{equation}
\label{eq:bdmu}
\mu_{j}^{-1} \lesssim \left\vert 4\xi \sin^2\left( \frac{\pi(j-1)}{2N}\right)-4\xi \right\vert^{-1} \lesssim \left\lvert \cos\left( \frac{\pi(j-1)}{2N}\right)\right\rvert^{-2}= \mathcal{O}(N^2 ). 
\end{equation}
Note that the last equality comes from the leading order in Taylor expansion.
Using \eqref{eq:bdmu} for $\mu_j^{-2}$ and also the explicit form of the eigenvectors \eqref{eq:eigenvectors}, yields that
\begin{equation}
\label{eq:start1}
\sum_{j=1}^{N-1} \frac{ \vert v_j(1) \vert^2}{\mu_{j}} \lesssim \frac{2}{N} \sum_{j=1}^{N-1} \left\lvert \cos\left(\frac{\pi (j - 1)}{2N}\right)\right\rvert^2 \left\lvert \cos\left( \frac{\pi(j-1)}{2N}\right)\right\rvert^{-2}  =  \mathcal O(1).
\end{equation}

We also record that again by \eqref{eq:bdmu} and \eqref{eq:eigenvectors}
\begin{equation}
\begin{split}
\label{eq:start2}
\sum_{j=1}^{N-1} \frac{ \vert v_j(1) \vert^2}{\mu_{j}^2} &\lesssim \frac{1}{N} \sum_{j=1}^{N-1} \left\lvert \cos\left(\frac{\pi (j - 1)}{2N}\right)\right\rvert^2 \left\lvert \cos\left( \frac{\pi(j-1)}{2N}\right)\right\rvert^{-4} \\&  \lesssim  \left\vert \operatorname{cos} \left(\frac{\pi}{2} - \frac{\pi}{N} \right)  \right\vert^{-2} = \mathcal O(N^2)
\end{split}
\end{equation}
where the last estimate follows by Taylor expanding around $\pi/2$. 
Since $\cos(\pi k-x) =(-1)^k \cos(x),$ we observe that also $v_j(1) = (-1)^{j-1} v_j(N)$ as

\[\cos \left(\tfrac{\pi(j-1)}{2N} \right) = \cos\left(\pi (j-1)- \tfrac{\pi (j-1)(N-\tfrac{1}{2})}{N} \right) = (-1)^{j-1} \cos\left(\tfrac{\pi (j-1)(N-\tfrac{1}{2})}{N} \right).\]

Let the rescaled Wigner $R$-matrix be defined using \eqref{eq:Wigner} as $$R_I(\lambda):=\widetilde R_I(\lambda+\lambda_N).$$ 
To make the sums on the right of \eqref{eq:Wigner} more transparent, we define matrices
\begin{equation}
\Gamma_{[1]}(j-1)=1 \text{ and }\Gamma_{\left\{1,N \right\}}(j-1) = \left(\begin{matrix} \gamma_1 &(-1)^{j-1} \sqrt{\gamma_1\gamma_N} \\ (-1)^{j-1}\sqrt{\gamma_1\gamma_N} & \gamma_N  \end{matrix} \right).
\end{equation}
These matrices allow us to rewrite the rescaled Wigner $R$-matrix in the more compact form
\begin{equation}
R_I(\lambda) = \sum_{ j=1}^N (\lambda-\mu_j)^{-1}  \vert v_j(1) \vert^2 \Gamma_{I}(j-1).
\end{equation}
Let us now restrict to the case $I= \left\{1,N \right\}$ and reduce $R_I \in \mathbb C^{2 \times 2}$ to a scalar equation (if there is friction at one end only, the Wigner $R$-matrix $R_I$ is already scalar). We then find that vectors $u=\frac{(\sqrt{\gamma_1},\sqrt{\gamma_N})^T}{\sqrt{\gamma_1+\gamma_N}}$ are eigenvectors to matrices $\Gamma_I$ such that $$\Gamma_{\{1,N\}}(j-1) u = (\gamma_1+\gamma_N)\delta_{j,\text{odd}} u$$ where $\delta_{j,\text{odd}}=1$ if $j$ is odd and $0$ otherwise. Similarly, for $j$ even, we can use vectors $u=\frac{(\sqrt{\gamma_1},-\sqrt{\gamma_N})^T}{\sqrt{\gamma_1+\gamma_N}}$ instead.

\noindent
Without loss of generality, let $N$ be odd, and $u=\frac{(\sqrt{\gamma_1},\sqrt{\gamma_N})^T}{\sqrt{\gamma_1+\gamma_N}}.$ It follows from \eqref{eq:eigenvectors} and Taylor expansion that $2\vert v_N(1)\vert^2 = \mathcal O(N^{-3}).$
We now use the expansion 
\begin{equation}
\label{eq:geometricseries}
(\lambda-\mu)^{-1}  = -\mu^{-1} \sum_{n=0}^{\infty}  \left(\lambda \mu^{-1}\right)^{n} = -\mu^{-1}- \mu^{-2}\lambda  -\mu^{-2} \lambda^2(\mu-\lambda)^{-1} 
\end{equation}
to rewrite the equation $(R_I(\lambda)-i)u=0$ in terms of scalar functions 
\begin{equation}
\begin{split}
\label{eq:f,g}
f(\lambda) &= \nu \lambda\quad \text{with}\quad  \nu :=- i-(\gamma_1+\gamma_N) \sum_{j=1,\ j \text{ odd}}^{N-2} \frac{\vert v_j(1) \vert^2}{\mu_j} \text{ and }\\
g(\lambda) &=(\gamma_1+\gamma_N)\left( \vert v_{N}(1) \vert^2 -\sum_{j=1,\ j \text{ odd}}^{N-2} \frac{\vert v_j(1) \vert^2\lambda^2}{\mu_j^2}- \sum_{j=1,\ j \text{ odd}}^{N-2}\frac{\vert v_j(1)\vert^2}{\mu_j^2}  \frac{\lambda^3}{\mu_j-\lambda} \right).
\end{split}
\end{equation}
Indeed, since 
\[(R_I(\lambda)-i)u = \left( \sum_{ j=1}^N (\lambda-\mu_j)^{-1}  \vert v_j(1) \vert^2 \delta_{j,\text{odd}} (\gamma_1+\gamma_N) - i \right) u \]
it follows by expanding $ (\lambda-\mu_j)^{-1} $, as in \eqref{eq:geometricseries}, and multiplying by $\lambda$ that
\begin{align*} \lambda ( R_I(\lambda)-i)u = (f(\lambda) + g(\lambda) ) u
\end{align*}
and thus we reduce our problem to a scalar one, since
\begin{equation}
\label{eq:equiv}
R_I(\lambda)u=iu\  \text{ if and only if }\  f(\lambda)+g(\lambda)=0.
\end{equation}
Let us now fix a ball
\begin{equation}
\label{eq:ball}
K:=B(0,r_N)\text{ for some }r_N\text{ to be determined.}
\end{equation}

We then find that for $\lambda \in \partial K$ we have for $f$ 
\begin{equation}
\label{eq:f}
r_N = \vert i \lambda \vert \le \vert f(\lambda) \vert \lesssim r_N.
\end{equation}

Equation \eqref{eq:start2} implies that for $\lambda \in \partial K$ we have for the second term in $g(\lambda)$, as in \eqref{eq:f,g}, that  
\[\left\lvert \lambda \right\rvert^2 \left\lvert \sum_{j=1}^{N-1} \tfrac{\vert v_j(1) \vert^2}{\mu_j^2} \right\rvert =\mathcal O(N^2 r_N^2).\]

Moreover, if we choose $r_N=\mathcal O(N^{-3})$, then by observing that 
by \eqref{eq:bdmu} and Taylor expansion $\cos(x) = -(x-\pi/2) + \mathcal O\left((x-\pi/2)^3\right)$
\begin{equation}
\begin{split}
\mu_j^{-1} &\lesssim\left\lvert \cos\left( \frac{\pi(j-1)}{2N}\right)\right\rvert^{-2} \lesssim  \left\vert \frac{\pi(j-1)}{2N} - \frac{\pi}{2} \right\vert^{-2}\lesssim  \left\vert \frac{\pi(N-1)}{2N} - \frac{\pi}{2} \right\vert^{-2}\\
& \lesssim \left\vert \frac{\pi}{2N} \right\vert^{-2}= \mathcal O(N^2). 
\end{split}
\end{equation}
Regarding the third term of $g(\lambda)$, as in \eqref{eq:f,g}, this one can also be estimated by 
\begin{equation}
\begin{split}
\left\lvert \sum_{j=1,\ j \text{ odd}}^{N-2}\frac{\vert v_j(1)\vert^2}{\mu_j^2}  \frac{\lambda^3}{\mu_j-\lambda} \right\rvert &= \vert \lambda \vert^2   \left\lvert \sum_{j=1,\ j \text{ odd}}^{N-2}\frac{\vert v_j(1)\vert^2}{\mu_j^2} \frac{1}{\frac{\mu_j}{\lambda}-1} \right\rvert \\&  \lesssim N^2 r_N^2 N^{-1}  = \mathcal{O}(N^{-5}). 
%\left\lvert \sum_{j=1,\ j \text{ odd}}^{N-2}\frac{\vert v_j(1)\vert^2}{\mu_j^3}  \frac{\lambda^3}{1-\frac{\lambda}{\mu_j}} \right\rvert\mathcal =O(r_N^2).
\end{split}
\end{equation}
since $\frac{\mu_j}{\lambda} \gtrsim \frac{N^{-2}}{N^{-3}}= \mathcal{O}(N)$ and so $ \frac{\lambda}{\mu_j-\lambda} = \mathcal{O}(N^{-1})$ .

This implies that for $\lambda \in \partial K$ we have 
\begin{equation}
    \begin{split}
    \label{eq:g}
(\gamma_1+\gamma_N)\left\lvert v_N(1) \right\rvert^2-\mathcal{O}(Nr_N^2)  - \mathcal O(N^2r_N^2)   &\le
 \left\lvert g(\lambda) \right\rvert\quad \text{and} \\
  (\gamma_1+\gamma_N)\left\lvert v_N(1) \right\rvert^2+ \mathcal{O}(Nr_N^2)  + \mathcal O(N^2r_N^2) &\ge | g(\lambda)|.
\end{split}
\end{equation}
\smallsection{Upper bound:}
Thus, we choose in \eqref{eq:ball} $r_N:=\frac{\alpha}{2} \left\lvert v_N(1) \right\rvert^2$ with $\alpha$ large enough, (but independent of $N$) such that 
% such that for example $ \alpha= \mathcal{O}\left( N^3 \right)$,
  together with \eqref{eq:f} and the upper bound in \eqref{eq:g}, they imply that on $\partial K$
%   \begin{align*}  \vert g(\lambda)\vert\  &\lesssim (\gamma_1+\gamma_N)N^{-3}+  \frac{N^6N^{-6}}{4}  =  (\gamma_1+\gamma_N)N^{-3}+ \frac{1}{4} \\ &\le \frac{1}{2} = \mathcal{O}(r_N) \sim \vert f(\lambda)\vert.
%\end{align*}

 \begin{align*}  \vert g(\lambda)\vert\  &\lesssim (\gamma_1+\gamma_N)N^{-3}+  \frac{\alpha^2 N^{-5}}{4} + \frac{\alpha^2 N^{-4}}{4}  \lesssim \mathcal{O}(r_N) \sim \vert f(\lambda) \vert
 \end{align*}
 which is the case if \begin{align*}  (\gamma_1+\gamma_N)N^{-3}+  \frac{\alpha^2 N^{-5}}{4} + \frac{\alpha^2 N^{-4}}{4}   &\lesssim  \frac{\alpha}{2} N^{-3}\quad \text{or}\quad (\gamma_1+\gamma_N) + \frac{\alpha^2}{4}( N^{-2 }+N^{-1} )  \lesssim  \frac{\alpha}{2} . 
 \end{align*}
 For large $N$ and large $\alpha$, the last inequality holds true.

 Therefore asymptotically with $N$, 
$$\vert f(\lambda)\vert > \vert g(\lambda)\vert\  \text{on}\ \partial K.$$
By Rouch\'e's Theorem, $f$ and $f+g$ have the same amount of zeros inside $K$. Since $f$ has precisely one root in $K$ at $\lambda=0$ so does $f+g$. 

This implies by the equivalence \eqref{eq:equiv} that $R(\lambda)u=iu$ has one solution $\lambda$ with $\lambda = \mathcal{O}(N^{-3})$ and so $\lambda_S = \mathcal O(N^{-3})$ which yields the upper bound on the spectral gap.

\smallsection{Lower bound:} The lower bound follows analogously. Assuming the spectral gap would decay faster than $\mathcal O(N^{-3})$, i.e. $\lambda_S/\vert v_N(1) \vert^2 = \smallO(1)$ then we can select $r_N=\vert v_N(1) \vert^2\smallO(1) $ in \eqref{eq:ball}. This way, $g(\lambda)$ does not have a root in $K:=B(0,r_N)$ by the lower bound in \eqref{eq:g}, as $\vert g(\lambda)\vert$ stays away from zero lower-bounded by a leading-order term $(\gamma_1+\gamma_N)\left\lvert v_N(1) \right\rvert^2$. Moreover, by the same lower bound in \eqref{eq:g} and upper bound in \eqref{eq:f} we find that on $\partial K$ we have for this choice of $r_N$
$$\vert g(\lambda)\vert > \vert f(\lambda)\vert\  \text{on}\ \partial K.$$
Thus, since $g$ does not have a root inside $K$, there is also no root to $f+g$ inside $K$ and thus by $\eqref{eq:equiv}$ we necessarily have that $N^{-3}\lesssim \lambda_S.$
\end{proof}
\begin{rem}[Dependence of $\lambda_S$ on the friction] We stress that our proof shows that the spectral gap depends on the friction constants $\gamma_1, \gamma_N,$ of the two terminal particles. In particular, by carefully analyzing this dependence in the proof, we see that there are constants $c_1,c_2>0$ so that $$c_1 \left(\frac{\gamma_1+\gamma_N}{1+\gamma_1+\gamma_N}\right) N^{-3} \le \lambda_S \le c_2(\gamma_1+\gamma_N) N^{-3}. $$
\end{rem}

\subsection{Higher-dimensional homogeneous networks}

We now turn to the $d$-dimensional homogeneous network of oscillators, on a square network for $d \ge 1$. We will show how we can extend ideas from the one-dimensional setting to the multi-dimensional case, in order to compute the spectral gap, by exploiting the separability of the Neumann Laplacian. 

%The number of the particles now is $N^d$ located at the edges of a square $[N]^d$ and, as before, the model shall be defined by the Hamilton function
%\begin{equation}
%\begin{split}
%H(\textbf{q},\textbf{p}) &= \frac{\langle \textbf{p},\textbf{m}_N^{-1}\textbf{p}\rangle}{2}+ V_{\eta,\zeta}(\textbf{q}) \text{ where } \\
%V_{{\bf \eta,\zeta}}(\textbf{q}) &= \sum_{i,j=1}^N \eta q_{ij}^2 + \sum_{i\sim j} \xi (q_i-q_{j})^2
%\end{split}
%\end{equation}
%where $i \sim j$ indicates now nearest neighbors.
Assuming $\eta$ and $\xi$ to be constant allows us to perform an analogous reduction of the high-dimensional spectral problem to a scalar problem, as in the one-dimensional case. We have a Schr\"odinger operator on $\mathbb{C}^{N^d}$ associated to the dynamics, as the first order part of the generator is expressed through the $2N^d\times 2N^d$ -dimensional matrix $\Omega_{[N]^d}$. 
%=\left( \begin{matrix}  \Gamma & -\textbf{m}_{[N]^d}^{-1} B_{N^d}^{1/2}\\ B_{N^d}^{1/2} \textbf{m}_{[N]^d}^{-1} & 0    \end{matrix} \right) .  $$  
The multi-dimensional Schr\"odinger operator has then the following spectral decomposition
\begin{equation}
\sqrt{B_{[N]^d}} = \sum_{i_1=1}^N \cdots \sum_{i_d=1}^N \lambda_{i_1 \cdots i_d}(\sqrt{B_{[N]^d}}) v_{i_1 \cdots i_d}^{\otimes d} 
\end{equation}
where $\lambda_{i_1 \cdots i_d}(\sqrt{B_{[N]^d}}) =\left( \sum_{k=1}^d \lambda_{i_k} \right)^{1/2}$ with $ \lambda_k= 4 \xi\sin^2\left(\frac{\pi ( k- 1)}{2N}\right)+\eta.$
The eigenvectors $v_{i_1 \cdots i_d}$ are the product states 
\begin{equation} \label{multi dim evectors}
v_{i_1 \cdots i_d}(j_1, j_2, \dots, j_d)= v_{i_1}(j_1)\cdots v_{i_d}(j_d)
\end{equation} 
such that 
\begin{equation}
v_j(i)= \begin{cases}
N^{- \frac{1}{2}} & , \ \mbox{j = 1}\\
\sqrt{\frac{2}{N}} \cos\left(\frac{\pi (j - 1)\left(i - \tfrac{1}{2}\right)}{N}\right) & , \ \mbox{otherwise}. 
\end{cases}
\end{equation}

\smallsection{$2$-particle friction on the $d$-dimensional network} As a first step we consider friction at two distinguished boundary particles out of the $N^d$. We now show how the method presented above applies if we consider friction at the two corners of the network, Fig. \ref{pic:networks1}, or at the centres of the two edges above and below, Fig. \ref{pic:networks3}. 

\begin{figure}
\centering
\begin{minipage}{0.40\textwidth}
\includegraphics[width=5cm]{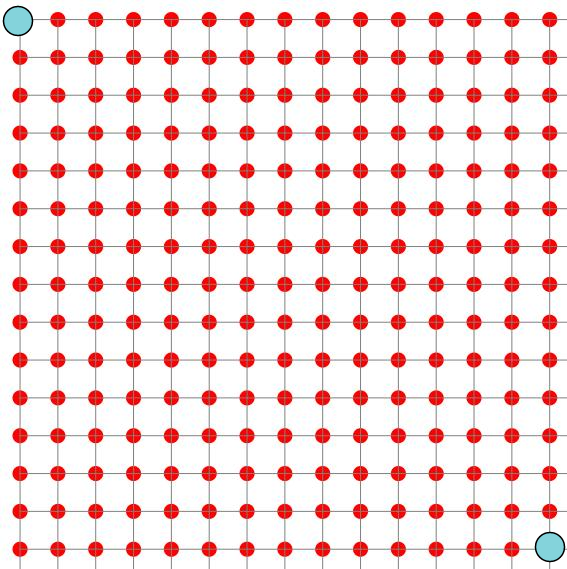}
\caption{Spectral gap $\sim N^{-6}$.}
 \label{pic:networks1}
\end{minipage}
\begin{minipage}{0.40 \textwidth}
\includegraphics[width=5cm]{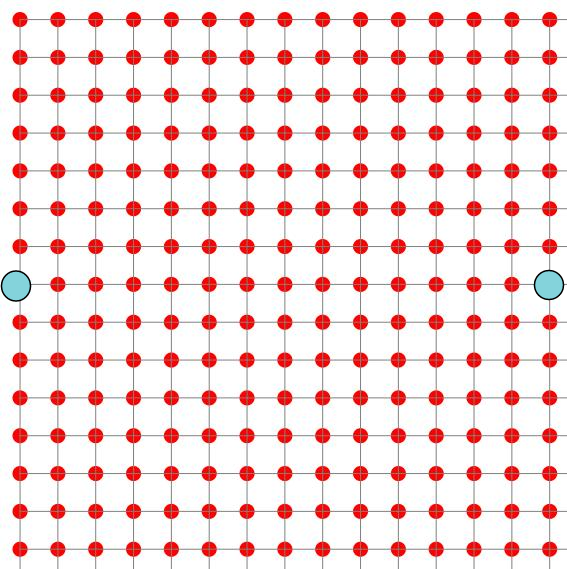} 
\caption{Spectral gap $\sim N^{-4}$.}
 \label{pic:networks3}
\end{minipage}
 \caption{The $\ZZ^2$-subnetwork with friction at the blue particles}
\end{figure}

We remind the high dimensional version of Lemma \ref{lemm:SGI}. We consider the matrix \begin{align} 
i \Omega_{[N]^d} = \left(  \begin{matrix} 0 & -iB_{[N]^d} \\ iB_{[N]^d} &0
\end{matrix} \right)   +   i \Gamma \oplus 0_{\mathbb C^{N^d \times N^d}}
\end{align}  
and reduce the high dimensional spectral problem for $\Omega_{[N]^d}$ to a lower dimensional spectral problem for the Wigner $\widetilde{R}_I$-matrix in $\mathbb{C}^{2 \times 2}$. From this lemma 
we get  the following representation of the $\widetilde{R}_I$-matrix: 
%{\color{blue}
%\begin{align}
% \widetilde{R}_I(\lambda) = \sum_{i,j}^N (\lambda - \lambda_{ij})^{-1} \sum_{\pm}  \sum_{(k_1,\dots,k_d), (\lambda_1,\dots, \lambda_d) \in I} \sqrt{\gamma_{k_1 k_2}\gamma_{\lambda_1 \lambda_2}}  \langle V_{ij}^{\pm}, e_{k_d}^{2N^d} \rangle \langle e_{N^2-k_2}^{2N^d},V_{ij}^{\pm}\rangle e_{k_d}^{\vert I\vert} \otimes e_{N^2-k_d}^{\vert I\vert}
% \end{align}}
 \begin{align}
 \widetilde{R}_{I}(\lambda)= \sum_{ j=1}^{N^d} (\lambda-\lambda_j)^{-1}  \sum_{\pm}\sum_{i_1,i_2 \in I} \sqrt{\gamma_{i_1}\gamma_{i_2}}\langle V_j^{\pm}, e_{i_1}^{2N^d} \rangle \langle e_{i_2}^{2N^d},V_j^{\pm}\rangle e_{i_1}^{\vert I\vert} \otimes e_{i_2}^{\vert I\vert}
 \end{align}
where $V_{j}^{\pm}=\tfrac{1}{\sqrt{2}}( v_{j},\pm i v_{j})^T$ are of the product form \eqref{multi dim evectors}
 with $v_{j}$ being the eigenvectors of $B_{[N]^d}$.

% in dimension $2$ for example $I= \{(1,1), (N,N)\}$. In this case, by repeating almost exactly the calculations of the Proposition \ref{prop:hom} we get that  the spectral gap scales like $N^{-6}$. In general we have the following Proposition. 

\begin{prop}[2-particle friction in homogeneous networks] \label{prop: hom d=2,3}
Let the dimension of the network be $d \ge 1$ and all $\eta_i>0$, $\xi_i>0$, and masses $m_i$ coincide, respectively. We consider two different scenarios: \begin{itemize}
\item  First, we assume that the two particles located at $(1,\dots,1),(N,\dots,N)$ are subject to non-zero friction and diffusion. The spectral gap of the harmonic network of oscillators satisfies $$N^{-3d}\lesssim \lambda_S \lesssim N^{-3d}.$$
\item Second, we assume the friction and diffusion acts on the particles located in the centre of the two edges of the network at $(1,  \lceil N/2 \rceil, \dots,  \lceil N/2 \rceil), (N,  \lceil N/2 \rceil, \dots,  \lceil N/2 \rceil)$. Then the spectral gap of the network of oscillators satisfies $$N^{-3-(d-1)}\lesssim \lambda_S \lesssim N^{-3-(d-1)}.$$
\end{itemize}
\end{prop}

\begin{proof} To keep the notation simple, we restrict us to stating the proof for $d=2$, only and we write for the eigenvalues $\lambda_{ij}(\sqrt{B_{N^2}}) =\left( \lambda_{i} +\lambda_j \right)^{1/2}$ and for 
the eigenvectors $v_{ij}$ which are the product states $
v_{ij}(k,l)= v_{i}(k)v_{j}(l)$. 
As in the one-dimensional case, we compute 
%{\color{blue} The notation seems to be messed up here. We have $\lambda_k$ and $\mu_{kj}$ but no $\lambda_{ij}.$} ( I Meant \lambda_{ij} before instead of \mu{ij} ) 
\begin{align}
| \lambda_{NN}(\sqrt{B_{N^2}})- \lambda_{ij}(\sqrt{B_{N^2}})| \gtrsim |(8\xi+ 2\eta)^{1/2}- \lambda_{ij}(\sqrt{B_{N^2}})|,
\end{align}
we define $\mu_{ij}:= \lambda_{ij}(\sqrt{B_{N^2}})- \lambda_{NN}(\sqrt{B_{N^2}}) $ so that  
\begin{align} \mu_{ij}^{-1}  &\lesssim | \lambda_{ij}^2-\lambda_{NN}^2|^{-1} \lesssim | (8\xi+2\eta)- \lambda_{ij}^2|^{-1} \\ &=  \left\vert 8\xi -4\xi \left( \left(\operatorname{sin} \left(\frac{\pi(j-1)}{2N}\right)\right)^2 + \left(\operatorname{sin}\left(\frac{\pi(i-1)}{2N}\right)   \right)^2  \right)\right\vert^{-1} \\ &\lesssim \left\vert \left(\operatorname{cos} \left(\frac{\pi(j-1)}{2N}\right)\right)^2 + \left(\operatorname{cos} \left(\frac{\pi(i-1)}{2N}\right)\right)^2 \right\vert^{-1} \\ &\lesssim \left\vert  \left(\operatorname{cos} \left(\frac{\pi(N-2)}{2N}\right)\right)^2   \right\vert^{-1} = \mathcal{O}(N^2)
\end{align}
where in the last line we Taylor expanded around $\pi/2$. We rescale the Wigner matrix $ R_I(\lambda)= \widetilde{R}_I (\lambda+ \lambda_{NN}) $ and write \begin{align}
R_I(\lambda)= \sum_{i,j}^N (\lambda-\mu_{ij})^{-1} |v_{ij}(1,1)|^2 \Gamma_I \in \mathbb{C}^{2\times 2}
\end{align}
where $$\Gamma_I:= \left( \begin{matrix} \gamma_{11}& (-1)^{i+j-2}\sqrt{\gamma_{11}\gamma_{NN}} \\ (-1)^{i+j-2}\sqrt{\gamma_{11}\gamma_{NN}} &  \gamma_{NN}  \end{matrix} \right) $$ since $ v_{ij}(1,1) =(-1)^{i+j-2} v_{ij}(N,N)$. 
Note that for $i+j=\text{even}$, the vector $u=(\gamma_{11}+\gamma_{NN})^{-1/2}(\sqrt{\gamma_{11}}, \sqrt{\gamma_{NN}})^T$ is an eigenvector to $\Gamma_I$: $$ \Gamma_I u = (\gamma_{11} + \gamma_{NN}) \delta_{i+j: even} u,$$ where we use the same notation as in the proof for the one dimension. We focus without loss of generality on this case only, since the remaining scenarios can be treated similarly. With the above formula and by expanding the term $ (\lambda- \mu_{ij})^{-1}$ we are able to rewrite the equation $(R_I(\lambda) - i)u=0$ in terms of two scalar functions $f,g$.  In particular $ \lambda(R_I(\lambda) - i)u = f(\lambda) + g(\lambda) $ with 
\begin{equation}
\begin{split}
f(\lambda) &=- i\lambda -(\gamma_{11}+\gamma_{NN}) \sum_{\substack{i,j=1,\\ i+j \text{ even}}}^{N-1} \lambda \frac{\vert v_{ij}(1,1) \vert^2}{\mu_{ij}}\quad \text{ and }\\
g(\lambda) &=(\gamma_{11}+\gamma_{NN})\left( \vert v_{NN}(1,1) \vert^2 -\sum_{\substack{i,j=1,\\ i+j \text{ even}}}^{N-1} \frac{\vert v_{ij}(1,1) \vert^2\lambda^2}{\mu_{ij}^2}- \sum_{\substack{i,j=1,\\ i+j \text{ even}}}^{N-1}\frac{\vert v_{ij}(1,1)\vert^2}{\mu_{ij}^2}  \frac{\lambda^3}{\mu_{ij}-\lambda} \right).
\end{split}
\end{equation}

We fix a ball $K:= B(0,r_N)$ and we estimate the following terms on the boundary $\partial K$: 
\begin{align}
\left\vert \sum_{\substack{i,j=1,\\ i+j \text{ even}}}^{N-1} \frac{\lambda \vert v_{ij}(1,1) \vert^2}{\mu_{ij}} \right\vert &= |\lambda| \left\vert \sum_{\substack{i,j=1,\\ i+j \text{ even}}}^{N-1}  \frac{ |v_i(1)|^2 |v_j(1)|^2}{\mu_{ij}} \right\vert \\ &\lesssim    \sum_{i,j=1}^{N-1}  N^{-2} \frac{r_N \left\vert \operatorname{cos}\left( \frac{\pi(i-1)}{2N} \right) \right\vert^2  \left\vert \operatorname{cos}\left( \frac{\pi(j-1)}{2N} \right) \right\vert^2   }{\left\vert \operatorname{cos}\left( \frac{\pi(i-1)}{2N} \right) \right\vert^2} = \mathcal{O}(r_N)
 %\sum_{i,j=1}^{N-1}  N^{-2} \frac{r_N N^{-2}N^{-2}}{N^{-2}} = \mathcal{O}(N^{-2}r_N)
\end{align} after Taylor expansions to estimate the norm of the eigenvectors. 
 Also \begin{align}
|\lambda|^2 \left\vert \sum_{\substack{i,j=1,\\ i+j \text{ even}}}^{N-1} \frac{ \vert v_{ij}(1,1) \vert^2}{\mu_{ij}^2} \right\vert &\lesssim |\lambda|^2 \sum_{\substack{i,j=1,\\ i+j \text{ even}}}^{N-1} N^{-2} \frac{ \left\vert \operatorname{cos}\left( \frac{\pi(j-1)}{2N} \right) \right\vert^2}{ \left\vert \operatorname{cos}\left( \frac{\pi(i-1)}{2N} \right) \right\vert^2} \lesssim \mathcal{O}(N^2r_N^2)
 \end{align}
 and 
 \begin{align} 
 \left\lvert \sum_{\substack{i,j=1,\\ i+j \text{ even}}}^{N-1}\frac{\vert v_{ij}(1,1)\vert^2}{\mu_{ij}^2}  \frac{\lambda^3}{\mu_{ij}-\lambda} \right\rvert &= \vert \lambda \vert^2   \left\lvert \sum_{\substack{i,j=1,\\ i+j \text{ even}}}^{N-1}\frac{\vert v_{ij}(1,1)\vert^2}{\mu_{ij}^2} \frac{1}{\frac{\mu_{ij}}{\lambda}-1} \right\rvert = \mathcal{O}(N^4 r_N^3)
 \end{align}
since $\frac{\mu_{ij}}{\lambda}\gtrsim N^{-2}r_N^{-1}$. Therefore we collect the following bounds for $f,g$: \begin{align} 
|f(\lambda)| =\mathcal{O}(r_N), &\quad |g(\lambda)| \gtrsim (\gamma_{11}+\gamma_{NN}) |v_{NN}(1,1)|^2 -  \mathcal{O}(N^2 r_N^2) -  \mathcal{O}(N^4r_N^3) \\ |g(\lambda)| &\lesssim (\gamma_{11}+\gamma_{NN}) |v_{NN}(1,1)|^2 +  \mathcal{O}(N^2 r_N^2)  +  \mathcal{O}(N^4r_N^3) 
\end{align} 
Choosing $r_N= \frac{\alpha}{2} |v_{NN}(1,1)|^2 = \mathcal{O}(N^{-6})$ gives the upper bound for the spectral gap, as in the end of the previous proof as well and the lower bound follows with the same contradiction argument.  

As regards the second part of the statement, \textit{i.e.} when the particles subject to friction are located in the centre of the bordered edges, \textit{i.e.} $ I=\{(1, \lceil N/2 \rceil ), (N, \lceil N/2 \rceil ) \}$, of the network rather than at the corners.  The proof follows exactly in the same way as in the first scenario, apart from the last part of it when we fix the radius $r_N$ of the ball $K$ in order to apply Rouch\'e's Theorem. In this case, taking $$r_N = \frac{\alpha}{2} \vert v_{NN}(1, \lceil N/2 \rceil) \vert^2 = \mathcal{O}(N^{-3}N^{-1}) $$ then immediately implies the result. 
\end{proof} 

\smallsection{$2N$-particles exposed to friction on two opposite edges}
As a second step we consider the most physically relevant case in higher dimensions, \textit{i.e.} we assume the friction to be imposed to all the particles located on the top edge of the network and on the bottom edge as well, cf. Fig. \ref{fig:relevant}.
We use the same techniques and notation as above and we will show how the same method applies to give an upper bound on the spectral gap. Thus here 
\[I=\{(1,1),\dots,(1,N),(N,1), \dots, (N,N)\}\text{ and }|I|=2N\] and the rescaled Wigner $R_I$-matrix in $\mathbb{C}^{2N\times 2N}$ will be 
 \begin{align}
 \widetilde{R}_{I}(\lambda)= \sum_{ j=1}^{N^d}  \sum_{\pm}\sum_{i_1,i_2 \in I} \frac{\sqrt{\gamma_{i_1}\gamma_{i_2}}}{\mu_{j} } \langle V_j^{\pm}, e_{i_1}^{2N^d} \rangle \langle e_{i_2}^{2N^d},V_j^{\pm}\rangle e_{i_1}^{\vert I\vert} \otimes e_{i_2}^{\vert I\vert}
 \end{align}
 where $\mu_{j}= \lambda_j-\lambda_{N^d}$.

Note that since in this case the Wigner matrix is still high-dimensional $2N \times 2N$ we shall support our analytical findings by some numerics too. In particular we have the following analytical result:
\begin{figure}
\centering
\includegraphics[width=5cm, angle=90]{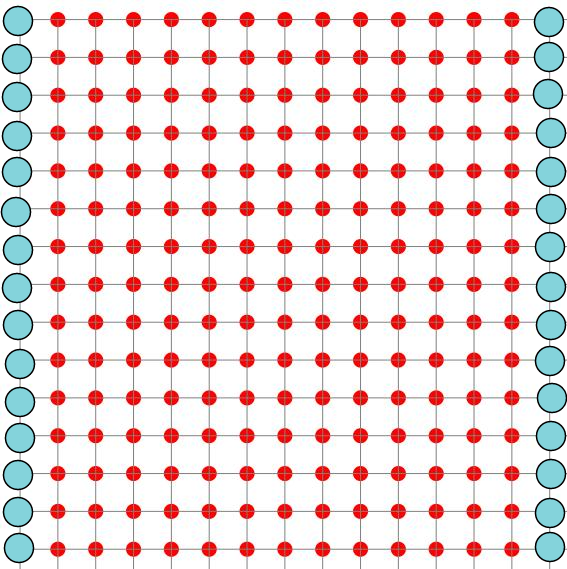}
 \caption{The $\ZZ^2$-subnetwork with friction at the blue particles on opposite edges.}
 \label{fig:relevant}
\end{figure}
\begin{prop}[$2N$-particle friction in homogeneous networks] \label{prop2: hom d=2,3}
\label{prop:hom2}
Let the two-dimensional square network graph with particles on the $N^2$ vertices, and all $\eta_i>0$, $\xi_i>0$, and masses $m_i$ coincide, respectively.
We assume that the $2N$ particles located at $$ \{(1,1),\dots,(1,N),(N,1), \dots, (N,N)\}$$ are subject to non-zero friction and diffusion. The spectral gap of the harmonic network of oscillators then satisfies $$ \lambda_S \lesssim N^{-5/2}.$$
\end{prop}
\begin{proof} We write again $\lambda_{ij}=\lambda_{ij}(\sqrt{B_{N^2}}) =\left( \lambda_{i}+\lambda_j  \right)^{1/2}$ and $v_{ij}$ to be the product states.  Using the equivalence of Lemma \ref{lemm:SGI}, we study the equation \begin{align} \label{eq: reduced prob} \operatorname{det}(F(\lambda) + G(\lambda) )=0\end{align}  in terms of the vectors $V_{ij}^{\pm}=\frac{1}{\sqrt{2}}(v_{ij}, \pm i v_{ij})^T$ and $\mu_{ij}= \lambda_{ij}- \lambda_{NN}$, where $\mu_{ij}^{-1} \lesssim \mathcal{O}(N^2)$ as before.  Note that we do not reduce our problem to a scalar one as in the $2$-particle friction cases above and thus we work with the matrix valued version of Rouch\'e's Theorem stated in Lemma \ref{lemm:Rouch}. Following the same method, the matrices $F(\lambda), G(\lambda)$ are defined as follows 
\begin{equation} 
\label{eq:F}
F(\lambda):= -i \lambda - \lambda \sum_{\substack{i,j=1\\ \text{or}\ i=N, j=N\\ \text{when}\ j\ne i}}^{N-1} \sum_{\pm}\sum_{\substack{i_1,i_2\in I}} \frac{ \sqrt{\gamma_{i_1}\gamma_{i_2}} }{\mu_{ij}} \langle V_{ij}^{\pm}, e_{i_1}^{2N^2}\rangle \langle e_{i_2}^{2N^2} ,V_{ij}^{\pm}\rangle e_{i_1}^{2N} \otimes e_{i_2}^{2N}
\end{equation} and 
\begin{equation}
\begin{split}
G(\lambda)&:=\underbrace{\sum_{\pm}\sum_{\substack{i_1,i_2 \in I}}  \sqrt{\gamma_{i_1}\gamma_{i_2}}  \langle V_{NN}^{\pm}, e_{i_1}^{2N^2}\rangle \langle e_{i_2}^{2N^2} ,V_{NN}^{\pm}\rangle e_{i_1}^{2N} \otimes e_{i_2}^{2N}}_{=:(\text{I})} \\
 &\underbrace{ \ - \lambda^2 \sum_{i,j=1}^{N-1} \sum_{\pm}\sum_{\substack{i_1,i_2 \in I}} \frac{\sqrt{\gamma_{i_1}\gamma_{i_2}} }{\mu_{ij}^2} \langle V_{ij}^{\pm}, e_{i_1}^{2N^2}\rangle \langle e_{i_2}^{2N^2} ,V_{ij}^{\pm}\rangle e_{i_1}^{2N} \otimes e_{i_2}^{2N}}_{=:(\text{II})}
\\
  & \underbrace{- \lambda^3 \sum_{i,j=1}^{N-1} \sum_{\pm}\sum_{\substack{i_1,i_2 \in I}} \frac{\sqrt{\gamma_{i_1}\gamma_{i_2}}}{\mu_{ij}^2(\mu_{ij}-\lambda)} \langle V_{ij}^{\pm}, e_{i_1}^{2N^2}\rangle \langle e_{i_2}^{2N^2} ,V_{ij}^{\pm}\rangle e_{i_1}^{2N} \otimes e_{i_2}^{2N}}_{=:(\text{III})}
\end{split}
\end{equation}
so that a solution to \eqref{eq: reduced prob} corresponds to the desired eigenvalue. Before we fix a ball $K=B(0,r_N)$, we want to find an upper bound for the $\| \Theta_N\|$, where 
\begin{equation}
\label{eq:ThetaN}
\Theta_N:= \sum_{\pm}\sum_{\substack{i_1,i_2 \in I}}  \sqrt{\gamma_{i_1}\gamma_{i_2}}    \langle V_{NN}^{\pm}, e_{i_1}^{2N^2}\rangle \langle e_{i_2}^{2N^2} ,V_{NN}^{\pm}\rangle e_{i_1}^{2N} \otimes e_{i_2}^{2N} 
\end{equation} i.e. is the first term, $(\text{I})$, of $G$. 
 We will define then, this bound to be the radius $r_N$ of the ball $K$ and we will proceed as in the previous proofs.
To understand the dependence of $\| \Theta_N\| $ on $N$ we make the following observation: \\
Due to the symmetries of the eigenvectors, e.g. that  $ v_{ij}(1,1) =(-1)^{i+j-2} v_{ij}(N,N)$,  it suffices to check the scaling of the entries at the columns $1, \dots, \lceil N/2 \rceil$ and only above the main diagonal of the matrix. We estimate them  by Taylor expanding and in a similar manner as in the previous proofs.  For example for the $3$ entries in the corners of the territory that we examine we  have \begin{align*}
&|v_{NN}(1,1)|^2 \lesssim N^{-6},\quad |v_{NN}(1,\lceil N/2 \rceil )|^2 \lesssim N^{-4},\\  &v_{NN}(1, \lceil N/2 \rceil ) v_{NN} (1,1) \lesssim N^{-6}+N^{-4}=\mathcal{O}(N^{-4}) 
\end{align*} by Young's inequality. 
 So all the entries scale at least like $N^{-4}$ which implies that $$ \|\Theta_N\| \le N^{1/2} \| \Theta_N\|_{\infty} \lesssim N^{1/2}NN^{-4} = \mathcal{O}(N^{-5/2}).$$ 

We now fix a ball $K=B(0,r_N)$ and choose the radius $r_N:= \frac{\alpha}{2} N^{-5/2}$.  
 Therefore it suffices to find a root of \eqref{eq: reduced prob} inside the ball $K$ and conclude the existence of an eigenvalue by Rouch\'e's theorem. We easily see that for all $v \neq 0$, on $\partial K$,   $\|F(\lambda) v \| \ge |\lambda| \|v \| = r_N \|v\|= \tfrac{\alpha}{2}N^{-5/2}\|v\| $ since the second term of the right hand side of \eqref{eq:F} is symmetric. We also collect the estimates for $(\text{II})$
\begin{align*}
&|\lambda|^2 \left\vert \sum_{\substack{i,j=1}}^{N-1} \frac{ \vert v_{ij}(1,1) \vert^2}{\mu_{ij}^2} \right\vert  \lesssim r_N^2 = \mathcal{O}(N^{-5}),\quad  |\lambda|^2 \left\vert \sum_{\substack{i,j=1}}^{N-1} \frac{ \vert v_{ij}\left(1,\lceil N/2 \rceil \right) \vert^2}{\mu_{ij}^2} \right\vert  \lesssim r_N^2= \mathcal{O}(N^{-5})\\  
&|\lambda|^2 \left\vert \sum_{\substack{i,j=1}}^{N-1} \frac{  v_{ij}(1,1) v_{ij}(1, \lceil N/2 \rceil )}{\mu_{ij}^2} \right\vert  \lesssim N^{-5}.
%\quad  |\lambda|^2 \left\vert \sum_{\substack{i,j=1}}^{N-1} \frac{v_{i,j}(1,N) v_{ij}(1, \lceil N/2 \rceil )}{\mu_{ij}^2} \right\vert  \lesssim N^{-6} 
\end{align*}
Moreover since $ \frac{\lambda}{\mu_{ij}-\lambda} = \mathcal{O}(N^{-1})$, for $(\text{III})$: 
\begin{align*}
 &\left\lvert \sum_{\substack{i,j=1}}^{N-1}\frac{\vert v_{ij}(1,1)\vert^2}{\mu_{ij}^2}  \frac{\lambda^3}{\mu_{ij}-\lambda} \right\rvert \lesssim N^{-5-1},\quad      \left\lvert \sum_{\substack{i,j=1}}^{N-1}\frac{\vert v_{ij}\left(1,\lceil N/2 \rceil \right) \vert^2}{\mu_{ij}^2}  \frac{\lambda^3}{\mu_{ij}-\lambda} \right\rvert \lesssim N^{-5-1},\\ & 
\left\lvert \sum_{\substack{i,j=1}}^{N-1}\frac{v_{ij}(1,1) v_{ij}(1, \lceil N/2 \rceil )}{\mu_{ij}^2}  \frac{\lambda^3}{\mu_{ij}-\lambda} \right\rvert \lesssim N^{-5-1}.
%\quad  |\lambda|^2 \left\vert \sum_{\substack{i,j=1}}^{N-1} \frac{v_{i,j}(1,N) v_{ij}(1, \lceil N/2 \rceil )}{\mu_{ij}^2} \right\vert  \lesssim N^{-6} 
\end{align*} So we can see that all the entries in $(\text{II})$ and $(\text{III})$ of $G$ scale like $\mathcal{O}(N^{-5})$ and $ \mathcal{O}(N^{-6})$ respectively. Thus, we find the following estimate on the operator norm of terms  $(\text{II})$ and $(\text{III})$
 \begin{align} 
 \|(\text{II})\| \le N^{1/2} \|(\text{II})\|_{\infty} \lesssim  N^{1/2} N N^{-5} = N^{1/2}N^{-4} = \mathcal{O}(N^{-7/2})
 \end{align}
 and 
 \begin{align} \|(\text{III})\| \le N^{1/2} \|(\text{III})\|_{\infty} \le N^{1/2} N N^{-6} = \mathcal{O}(N^{-9/2}).
\end{align} 
 We conclude that \begin{align}  \|G\|\lesssim N^{-5/2} + N^{-7/2} + N^{-9/2} = N^{-5/2}(1+ \smallO(1)) = \mathcal{O}(N^{-5/2}). 
 \end{align} 
 We choose $\alpha$ large enough so that we have $\| F(\lambda)v \| > \| G(\lambda) v\|$ on $\partial K$, for  all $v \neq 0$. Since $F(\lambda)$ is not invertible exactly at $0$ inside $K$, we have that there is one point inside $K$ so that $F(\lambda)+ G(\lambda) $ is not invertible or in other words there is one root of $R_I(\lambda)-iu=0$ with $\lambda \lesssim N^{-5/2}$. 
\end{proof} 

Proposition \ref{prop:hom2} provides only an upper bound on the spectral gap. The main obstruction to find sharp estimates on the spectral gap is to obtain precise asymptotics on the scaling of the operator norm, $\Vert \Theta_{N} \Vert$, in \eqref{eq:ThetaN}. By numerically calculating the operator norm of $\Vert \Theta_{N} \Vert$, we see that the optimal scaling is $\sim N^{-3}$ instead of $\mathcal O(N^{-5/2})$ as used in the proof of Proposition \ref{prop:hom2}. 
\begin{figure}
\centering
\includegraphics[width=9cm]{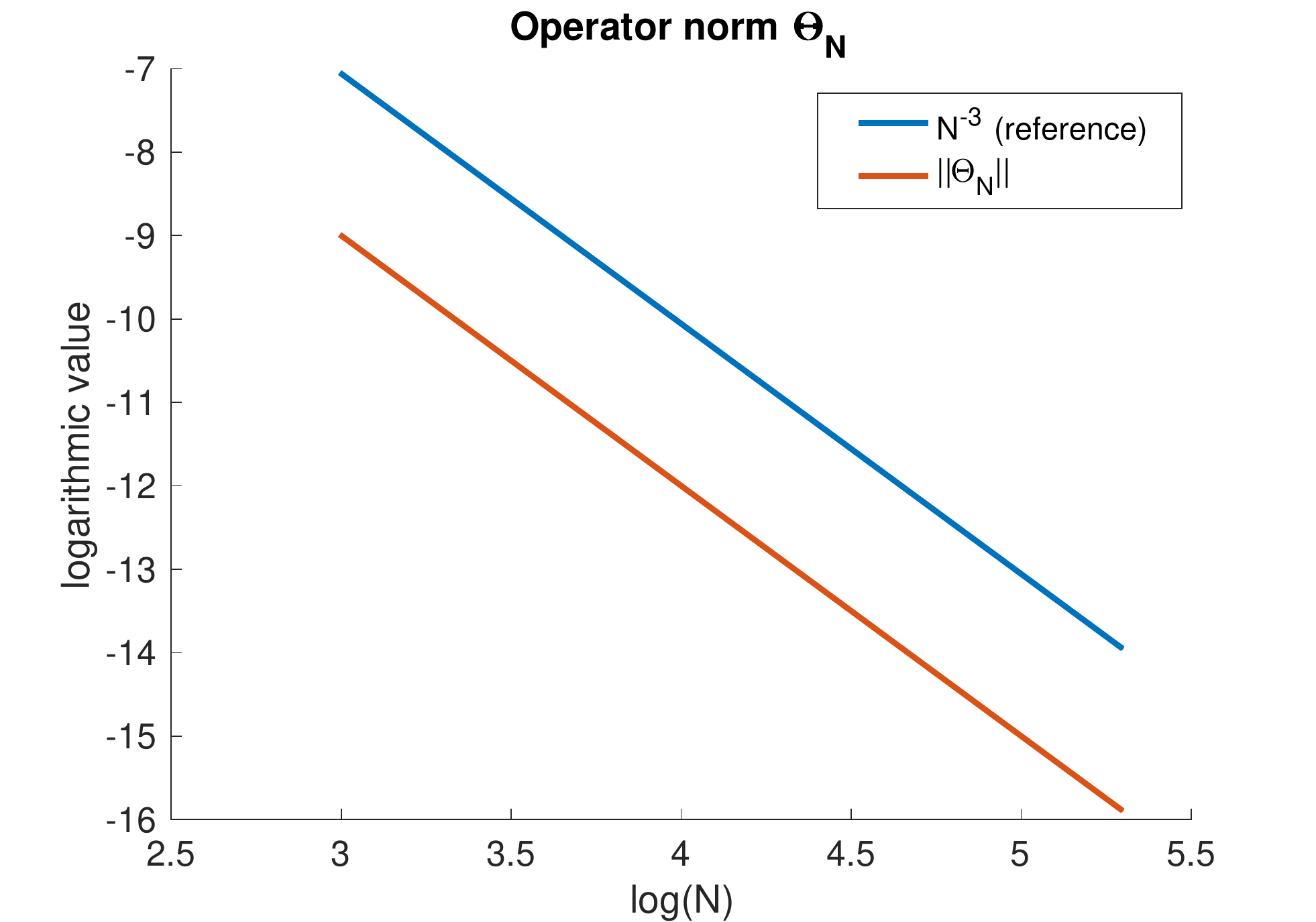}
\caption{Log-log plot of operator norm $\Vert \Theta_N \Vert$ and reference curve $N^{-3}.$}
\label{fig:opnorm}
\end{figure}

\subsection{Single impurities in the chain}
An impurity in the chain of oscillators refers to a particle with different physical properties from all the remaining particles. 
Since certain local impurities such as perturbations of the potential strength for a single particle, are finite-rank perturbations of the discrete Schr\"odinger operator, they do not effect the essential spectrum, but can lead to additional discrete spectrum in the limiting operator $B_{[\infty]^d}$.\\

To understand the eigenstates associated to certain points in the discrete spectrum better, we recall a classical result due to Combes and Thomas:
\begin{theo}
\label{theo:CombesThomas}
Let $V \in \ell^{\infty}(\ZZ^d)$ and suppose that $(-\Delta_{[\infty]^d}+ V)u = \lambda u$ with $\lambda \notin [0,4d]=:\Spec(-\Delta_{[\infty]^d}).$ 
If $\limsup_{\vert n \vert \rightarrow \infty} \vert V(n) \vert < \inf_{\mu \in \Spec(-\Delta_{[\infty]^d})} \vert \mu- \lambda \vert$,  then there is $\nu>0$ such that 
\[ u \in \left\{ \varphi \in \ell^2(\ZZ^d): \sum_{n \in \ZZ^d} \operatorname{exp}\left(2\nu (1+\vert n \vert^2)^{1/2} \right) \vert \varphi(n) \vert^2 < \infty \right\}. \]
\end{theo}

The above theorem implies that these eigenstates are exponentially localized in space and -as we show- will cause an exponentially fast closing of the spectral gap. This is in particular what happens if the pinning strength $\eta$ of a single particle is \textit{significantly weaker} than the pinning strength of all the other particles (the ''flying away'' particle). Note that in contrast to a weak pinning potential, a locally vanishing interaction potential would just decouple the chain into two independent pieces.  

Let $I_{m_0}^{[N]^d}:=\times_{i=1}^d \left\{m_0 - \lfloor{\tfrac{N-1}{2} \rfloor}e_i,..., m_0 + \lceil{\tfrac{N-1}{2} \rceil}e_i \right\}$ be a set of size $N^d$ around $m_0$. 
To switch from $\mathbb R^{[N]^d}$ to $\ZZ^d,$ we define the inclusion map $\iota: \mathbb R^{[N]^d} \to \ell^2(\mathbb Z^d)$ by 
\[ \left(\iota   x \right)\left(i\right) :=  \Bigg\{ \begin{array}{ c l } x(i),\ &\text{for }i \in I_{m_0}^{[N]^d}
\\ 0,\   &\text{otherwise} \end{array} \]
and define the restriction of the Schr\"odinger operator by
\[ B_{I_{m_0}^{[N]^d}}x:=B_{[\infty]^d} (\iota x).\]

\begin{lemm}
\label{lemm:jess}
For some $m_0 \in \ZZ^d$, let $B_{I_{m_0}^{[N]^d}}$ be a finite $[N]^d$-size truncation of a bounded discrete Schr\"odinger operator $B_{[\infty]^d}$ on $\ell^2(\mathbb Z^d).$ Let $\varphi$ be an eigenfunction to $B_{[\infty]^d}$ with eigenvalue $\lambda_{\infty}$ and assume that $\varphi$ is exponentially localized to a point $m_{0} \in \mathbb Z^d$ such that 
\begin{equation}
\label{eq:decayjessica}
\vert \varphi(n) \vert =\mathcal O(e^{-D \vert n-m_0 \vert}) \text{ for all } n \in \ZZ^d.
\end{equation}
We then define the finite $N^d$-size restriction 
\[\widetilde{\varphi}_{I_{m_0}^{[N]^d}}:=\frac{\varphi \vert_{I_{m_0}^{[N]^d}}}{\Vert  \varphi \vert_{I_{m_0}^{[N]^d}} \Vert}.\]
Furthermore, assume that the operator $B_{[N]^d}$ has a unique eigenvalue $\widetilde \lambda_N$, with associated eigenvector $\widetilde \psi_N$, such that $\inf_{\lambda \in \Spec(B_{[N]^d})} \vert \lambda_{\infty}-\lambda \vert=d(\lambda_{\infty}, \widetilde \lambda_N)$ and a spectral gap $\alpha_N>0$ such that 
\[\Spec(B_{[N]^d})\cap (\widetilde \lambda_N-\alpha_N, \widetilde \lambda_N+\alpha_N) = \{ \widetilde \lambda_N \},\] then
\[\Vert \widetilde{\psi}_N- \widetilde \varphi_N  \Vert = \mathcal O(e^{-DN}\alpha_N^{-1}).\]
\end{lemm}
\begin{proof}
We first record that \eqref{eq:decayjessica} implies the following exponential tail bound
\begin{equation}
\begin{split}
\label{eq:tail}
\sqrt{\sum_{\vert m - m_0 \vert \ge N/2} \vert \varphi(m)\vert^2}&= \mathcal O\left(e^{-D N}\right).
\end{split}
\end{equation}
We also define the infinite matrix $\widehat{B}_{I_{m_0}^{[N]^d}}$ given as the direct sum of operators
$$\widehat{B}_{I_{m_0}^{[N]^d}} :=B_{I_{m_0}^{[N]^d}}\oplus 0 $$
with respect to the direct sum decomposition $\ell^2(\mathbb Z^d) \simeq \ \ell^2(I_{m_0}^{[N]^d}) \oplus \ell^2(\ZZ^d \backslash I_{m_0}^{[N]^d}).$
Thus, we have
\begin{align*}
\frac{\Vert (\widehat{B}_{I_{m_0}^{[N]^d}}-\lambda_{\infty} ) \iota(\varphi \vert_{I_{m_0}^{[N]^d}} ) \Vert}{\Vert \iota (  \varphi \vert_{I_{m_0}^{[N]^d}})  \Vert} &\le \frac{\Vert ( \widehat{B}_{I_{m_0}^{[N]^d}}  -B_{[\infty]^d} )  \iota (  \varphi \vert_{I_{m_0}^{[N]^d}}) \Vert}{\Vert \iota (  \varphi \vert_{[N]})  \Vert} + \frac{\Vert ( B_{[\infty]^d}-\lambda_{\infty} I)   \iota (  \varphi \vert_{I_{m_0}^{[N]^d}}) \Vert  }{\Vert \iota (  \varphi \vert_{I_{m_0}^{[N]^d}})  \Vert}  \\ &\le     \frac{\Vert ( B_{[\infty]^d}-\lambda_{\infty} I ) ( \iota (  \varphi \vert_{I_{m_0}^{[N]^d}}) - \varphi ) \Vert}{\Vert \iota (  \varphi \vert_{I_{m_0}^{[N]^d}})  \Vert} +  \frac{\Vert ( B_{[\infty]^d}-\lambda_{\infty} I )   \varphi  \Vert}{\Vert \iota (  \varphi \vert_{I_{m_0}^{[N]^d}})  \Vert}    \\ & = \mathcal{O}(e^{-DN})
\end{align*}
where the first term on the right-hand side of the first line vanishes, up to exponentially small boundary terms, and in the last line we used the estimate \eqref{eq:tail} that holds for the eigenfunctions of $B_{\infty}$. 
Thus, the above bounds show that 
\begin{align}
\label{eq:spectraldecomp}
\frac{\Vert (B_{{I_{m_0}^{[N]^d}}}-\lambda_{\infty} ) \varphi \vert_{I_{m_0}^{[N]^d}} \Vert}{\Vert   \varphi \vert_{I_{m_0}^{[N]^d}}\Vert} = \mathcal{O}(e^{-DN})
\end{align}
and this implies by self-adjointness that also 
\begin{equation}
\label{eq:expclose}
\inf_{\lambda \in \Spec (B_{I_{m_0}^{[N]^d}})}\vert \lambda- \lambda_{\infty}\vert= \mathcal O(e^{-DN}).
\end{equation} 
That $\widetilde{\varphi}_{I_{m_0}^{[N]^d}}:=\frac{\varphi \vert_{I_{m_0}^{[N]^d}}}{\Vert  \varphi \vert_{I_{m_0}^{[N]^d}} \Vert} \in \mathbb R^{N^d}$ is exponentially close to an eigenvector $\widetilde\psi_N$ with eigenvalue $\widetilde\lambda_N$ of $B_{I_{m_0}^{[N]^d}}$ follows then by the spectral decomposition of $B_{I_{m_0}^{[N]^d}}$: In particular, let $(\psi_i)$ be an ONB of $B_{I_{m_0}^{[N]^d}}$ with eigenvalues $\lambda_i$ then we find by \eqref{eq:spectraldecomp} that 
\[\Vert (B_{I_{m_0}^{[N]^d}}-\lambda_{\infty} ) \widetilde{\varphi}_{I_{m_0}^{[N]^d}}\Vert= \sqrt{\sum_{i=1}^{N^d} \vert \langle \psi_i, \widetilde{\varphi}_{I_{m_0}^{[N]^d}} \rangle \vert^2 \vert \lambda_i - \lambda_{\infty} \vert^2} = \mathcal O(e^{-DN})=:\varepsilon. \]
This implies that for any $\nu>0$
\begin{equation}
\label{eq:pressurecond}
\sqrt{\sum_{i \in [N^d]: \vert \lambda_i - \lambda_{\infty} \vert \ge \nu \varepsilon} \vert \langle \psi_i, \widetilde{\varphi}_N \rangle \vert^2} \le \nu^{-1}. 
\end{equation}
Now, using that $\widetilde \lambda_N$ is a distance $\alpha_N$ apart from the rest of the spectrum of $B_{I_{m_0}^{[N]^d}}$ and $\lambda_{\infty}$ is exponentially close to $\widetilde \lambda_N$ by \eqref{eq:expclose} with some eigenvector $\widetilde \psi_N$ of $B_{I_{m_0}^{[N]^d}}$, we have from \eqref{eq:pressurecond} by setting $\nu:=\varepsilon^{-1}c \alpha_N$ that the coefficients of $\widetilde{\varphi}_{I_{m_0}^{[N]^d}}$ in the ONB with respect to all other eigenvectors of $B_{I_{m_0}^{[N]^d}}$ are exponentially small. Thus, we find that
\[\Vert \widetilde{\psi}_N- \widetilde \varphi_{I_{m_0}^{[N]^d}}  \Vert = \mathcal O(\nu^{-1}) = \mathcal O(e^{-DN}\alpha_N^{-1})\] such that the two vectors are exponentially close to each other.
\end{proof}\

\begin{prop}[Impurity]
\label{prop:impure}
Without loss of generality, let $N$ be an even number and consider a chain of oscillators with equal masses and unit coupling strength $\xi_i=1$. In addition, we assume that there is always at least one particle experiencing friction at the boundary and that the friction of particles is uniformly bounded in $N$. We define the centre point $c_d(N)=(N/2,..,N/2)$ and assume that $$\eta_{c_d(N)} +2d+\varepsilon \le \eta_i\quad  \text{uniformly in}\ [N]^d $$ for some $\varepsilon>0, i \neq c_d(N) .$ Then, the spectral gap of the harmonic chain of oscillators described by the operator \eqref{eq:generator} with the impurity described by the assumptions on the potentials given above, decays exponentially fast.
%We assume that the spectral gap $ \lambda_2(\sqrt{B_N})-\lambda_1(\sqrt{B_N})$ of the square-root of the discrete Schr\"odinger operator, decays at most sub-exponentially and the normalized ground-state $v_1$ of $B_N$ decays exponentially in the number of particles, i.e. for all $N$ there are $C,c>0$ such that $\vert v_1(1) \vert^2 ,\vert v_1(N) \vert^2 \le Ce^{-cN}.$
\end{prop}
\begin{proof}
First we show that the above assumptions imply the existence of an exponentially localized groundstate of $B_{[N]^d}$: 
\medskip

Let $V_{[N]^d}:= (\xi_i)_{i \in [N]^d}$, the min-max principle implies for the discrete Schr\"odinger operator \eqref{eq:Schroe} that
\[\lambda_1(B_{[N]^d})\le \lambda_1(V_{[N]^d}) + \langle e_{c_d(N)}, (-\Delta_{[N]^d})e_{c_d(N)} \rangle = \lambda_1(V_{[N]^d}) + 2d\]
where $e_{c_d(N)}$ is the unit vector that vanishes at every point different from $c_d(N).$
On the other hand, Weyl's inequalities and the assumptions on the coefficients of the pinning potential, imply that 

$$ \lambda_1(B_{[N]^d}) \le \| \Delta_{[N]^d}\| + \lambda_1(V_{[N]^d}) < \eta_{i \ne c_d(N)} - \epsilon= \lambda_2(V_{[N]^d}) - \epsilon \le \lambda_2(B_{[N]^d})- \epsilon $$
where $ \| \Delta_{[N]^d}\| \le 2d$ is the operator norm of the discrete Laplacian. 
%Thus, 
%\[\lambda_{1}(V_N) \le \lambda_1(B_N) + \lambda_2(V_N) \le \lambda_2(B_N).\]
Hence, $B_{[N]^d}$, and thus $\sqrt{B_{[N]^d}}$ has a spectral gap uniformly in $N$ since $$\lambda_1(B_{[N]^d})+\varepsilon \le \lambda_2(B_{[N]^d})\quad \text{uniformly in}\ N.$$ Now this implies that for some universal $c>0$ we have $\vert v_1(1) \vert^2,\vert v_1(N) \vert^2 \lesssim e^{-c N}$: from Theorem \ref{theo:CombesThomas}, cf. also \cite[Lemma $2.5$]{T},  we have that the ground state eigenfunction $u$ of the limiting operator $B_{[\infty]^d}$ is exponentially localized since the operators $B_{[N]^d}$ possess a uniform spectral gap of size at least $\alpha_N:=\varepsilon$ and $\lambda_1(B_{[N]^d}) \notin \Spec_{\operatorname{ess}}(B_{[\infty]^d}).$

The previous Lemma \ref{lemm:jess} then implies with $m_0=c_d(N)$ that there is an eigenstate $v_1$ to $B_{[N]^d}$
\[\Vert v_1- u\vert_{\varphi_{[N]^d}} \Vert = \mathcal O(e^{-DN/2}\varepsilon^{-1}).\]

To conclude the existence of an eigenvalue converging exponentially fast to zero, we shall restrict us again to the case $d= 2$ to keep the notation simple while at the same time dealing with all technicalities of the multi-dimensional setting.

Using the equivalence of Lemma \ref{lemm:SGI}, we study the equation \begin{align} \label{eq: reduced prob3} \operatorname{det}(F(\lambda) + G(\lambda) )=0\end{align} in terms of the vectors $V_{j}^{\pm}=\frac{1}{\sqrt{2}}(v_{j}, \pm i v_{j})^T$ and $\mu_{j}= \lambda_{j}- \lambda_{1}$, where $v_j$ are the eigenvectors of the Schr\"odinger operator $-\Delta_{[N]^2}+V_{[N]^2}$ with eigenvalue $\lambda_j$ and $\lambda_{1}:=\lambda_1(\sqrt{B_{[N]^d}})$.  The matrices $F(\lambda), G(\lambda)$ are then defined as follows 
\begin{equation} 
\label{eq:F3}
F(\lambda):= -i \lambda - \lambda \sum_{j=2}^{N^2} \sum_{\pm}\sum_{\substack{i_1,i_2\in I}} \frac{ \sqrt{\gamma_{i_1}\gamma_{i_2}} }{\mu_{j}} \langle V_{j}^{\pm}, e_{i_1}^{2N^2}\rangle \langle e_{i_2}^{2N^2} ,V_{j}^{\pm}\rangle e_{i_1}^{\vert I \vert} \otimes e_{i_2}^{\vert I \vert}
\end{equation} and 
\begin{equation}
\begin{split}
G(\lambda)&:=\underbrace{\sum_{\pm}\sum_{\substack{i_1,i_2 \in I}}  \sqrt{\gamma_{i_1}\gamma_{i_2}}  \langle V_{1}^{\pm}, e_{i_1}^{2N^2}\rangle \langle e_{i_2}^{2N^2} ,V_{1}^{\pm}\rangle e_{i_1}^{\vert I \vert} \otimes e_{i_2}^{\vert I \vert}}_{=:(\text{I})} \\
 &\underbrace{ \ - \lambda^2 \sum_{j=2}^{N^2} \sum_{\pm}\sum_{\substack{i_1,i_2 \in I}} \frac{\sqrt{\gamma_{i_1}\gamma_{i_2}} }{\mu_{j}^2} \langle V_{j}^{\pm}, e_{i_1}^{2N^2}\rangle \langle e_{i_2}^{2N^2} ,V_{j}^{\pm}\rangle e_{i_1}^{\vert I \vert} \otimes e_{i_2}^{\vert I \vert}}_{=:(\text{II})}
\\
  & \underbrace{- \lambda^3 \sum_{j=2}^{N^2} \sum_{\pm}\sum_{\substack{i_1,i_2 \in I}} \frac{\sqrt{\gamma_{i_1}\gamma_{i_2}}}{\mu_{j}^2(\mu_{j}-\lambda)} \langle V_{j}^{\pm}, e_{i_1}^{2N^2}\rangle \langle e_{i_2}^{2N^2} ,V_{j}^{\pm}\rangle e_{i_1}^{\vert I \vert} \otimes e_{i_2}^{\vert I \vert}}_{=:(\text{III})}
\end{split}
\end{equation}
so that a solution to \eqref{eq: reduced prob3} corresponds to the desired eigenvalue. Before we fix a ball $K=B(0,r_N)$, we want to find an upper bound for the $\| \Theta_N\|$, where 
\begin{equation}
\label{eq:ThetaN3}
\Theta_N:= \sum_{\pm}\sum_{\substack{i_1,i_2 \in I}}  \sqrt{\gamma_{i_1}\gamma_{i_2}}    \langle V_{1}^{\pm}, e_{i_1}^{2N^2}\rangle \langle e_{i_2}^{2N^2} ,V_{1}^{\pm}\rangle e_{i_1}^{2N} \otimes e_{i_2}^{2N} 
\end{equation} i.e. is the first term, $(\text{I})$, of $G$. 
From the exponential decay of the eigenstate $V_1^{\pm}$ it follows that for some $c>0$ we have 
\[ \Vert \Theta_N \Vert =\mathcal O(N e^{-cN}).\]

We now fix a ball $K=B(0,r_N)$ and choose the radius $r_N:= \mathcal O(Ne^{-cN})$.  
Therefore it suffices to find a root of \eqref{eq: reduced prob} inside the ball $K$ and conclude the existence of an eigenvalue by Rouch\'e's theorem. We easily see that for all $v \neq 0$, and  $\lambda \in \partial K$,   $\|F(\lambda) v \| \ge |\lambda| \|v \| = r_N \|v\|$ since the second term of the right hand side of \eqref{eq:F2} is symmetric. On the other hand, 
\begin{align*}
\Vert (\text{II}) \Vert = \mathcal O(N^2e^{-2cN}) \text{ and }\Vert (\text{III}) \Vert = \mathcal O(N^2e^{-3cN}).
\end{align*}

Thus, we have $\| F(\lambda)v \| > \| G(\lambda) v\|$ on $\partial K$, for all $v \neq 0$. Since $F(\lambda)$ is not invertible exactly at $0$ inside $K$, we have from Lemma \ref{lemm:Rouch} that there is one point inside $K$ so that $F(\lambda)+ G(\lambda) $ is not invertible or in other words there is one root of $R_I(\lambda)-iu=0$ with $\lambda \lesssim N e^{-cN}$. 
\end{proof} 
\subsection{Disordered chains}
We now study the case of a disordered pinning potential, i.e. we assume that $\eta_i>0$ are independent identically distributed (i.i.d.) random variables drawn drawn from some bounded density distribution $$\eta_i \sim \rho \in C_c(0,\infty).$$ Note that additional disorder in the interaction strengths leads to the-somewhat analogous study of random Jacobi operators which is for example discussed in \cite[Ch. $5$]{T}. In particular, localization for off-diagonal disorder in discrete Schr\"odinger operators, corresponding to random interactions in the chain of oscillators, is studied in \cite{DKS, DSS}. \\
Note that disordered harmonic chains have been studied before \cite{CL74, CL71}, even though in these works the randomness is posed in the masses of the particles, rather than the coefficients of the pinning potentials. However, the effect of localization does extend to that setting as well and can be studied- up to some technicalities- along the lines of the proof presented here.  We illustrate in Fig. \ref{fig:clement3} that all types of disorder yield an exponentially fast closing of the spectral gap.

\begin{figure}
\includegraphics[width=10cm]{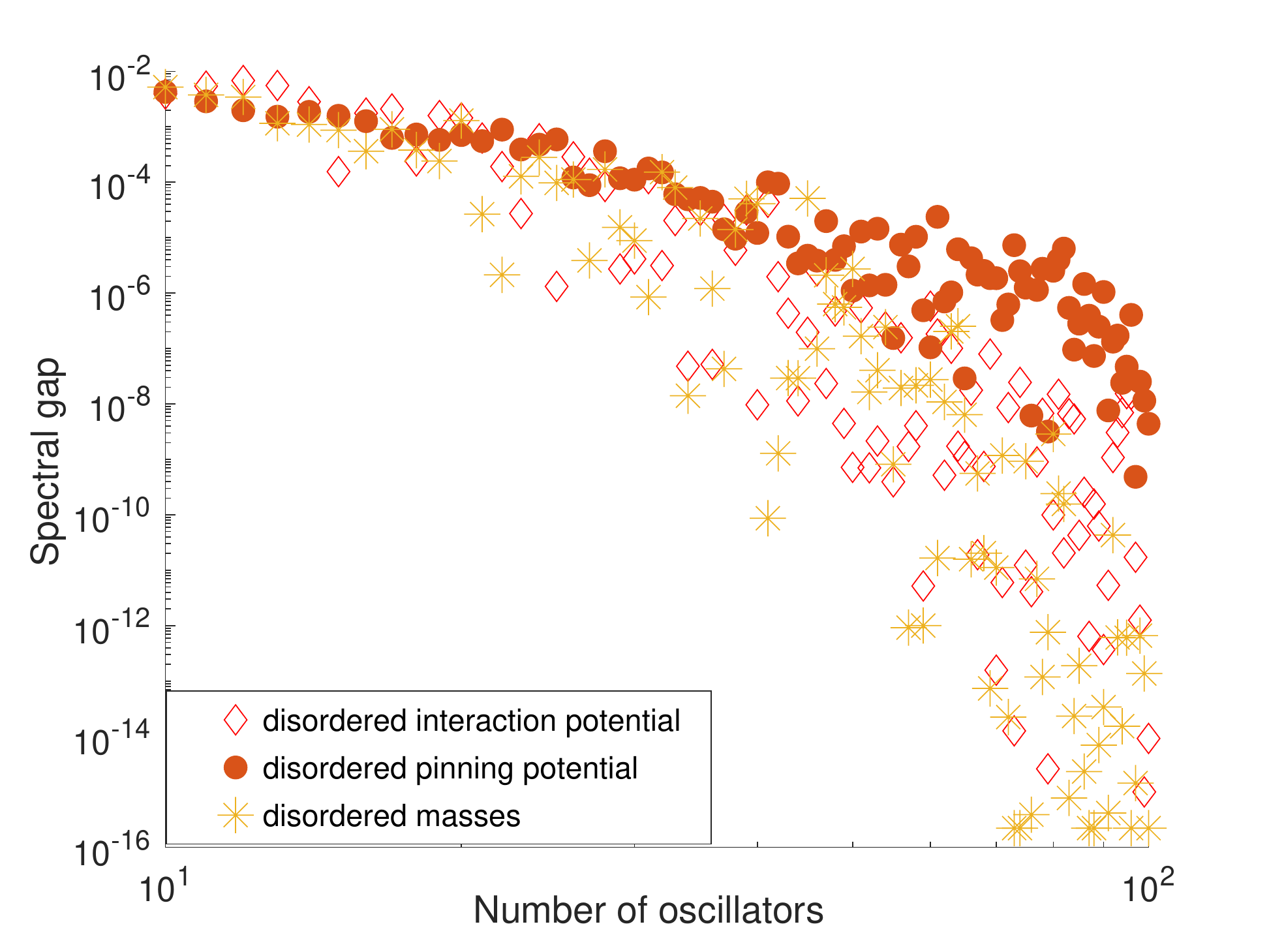}
\caption{Log-log plot of the spectral gap for the one-dimensional chain of oscillators and different types of disorder: Random masses $m_i=\frac{1}{1+X_i}$, random interaction $\xi_i = 1 +X_i$ and random pinning potential $\eta_i=1+X_i$ where $X_i \sim U[0,1]$ are uniform iid.}
\label{fig:clement3}
\end{figure}

The generator of the dynamics is the operator $\mathcal{L}$ given by \eqref{eq:generator}. Considering friction and diffusion at at least one end of the chain, cf. Proposition \ref{prop about reduction of sg}, the spectral gap is still given as \[\lambda_S:=\inf\{ \text{Re}(\lambda): \lambda \in \Spec( \Omega_{[N]^d}) \} .\]
From general results stated in Lemma \ref{lemm:SGI}, studying the spectrum of the matrix $\Omega_{[N]^d}$ is equivalent to studying the points at which the lower dimensional Wigner $\widetilde R_{I}$-matrix is not invertible. The matrix $B_{[N]^d}$, appearing in the matrix entries of $\Omega_{[N]^d}$ \eqref{eq:generator}, is the restriction to a finite domain of size $N^d$ of the one-dimensional \emph{discrete Anderson model}. This is explained below.\\

In the analysis of the disordered case it makes the analysis slightly simpler by labelling particles instead of $[N]^d$ rather by a set 
\[[\pm N]^d:=\{-N, -N+1,...,N-1,N\}^d,\]
i.e. we study the scaling of the spectral gap for $(2N+1)^d$ particles as a function of $N$ where we assume the chain to grow in all directions.

For disorder in the pinning potential, the limiting discrete Schr\"odinger operator $B_{[\infty]^d}$ is the multi-dimensional \emph{discrete Anderson model}:  the discrete Anderson model is a discrete Schr\"odinger operator with random single-site potential introduced by Anderson \cite{AN58} to describe the absence of diffusion in disordered quantum systems. It is the random discrete Schr\"odinger operator on $\ell^2(\mathbb Z^d)$
 $$ H^{[\infty]^d}_{\omega,\lambda}= -\Delta_{[\infty]^d} + \lambda V_{\omega}$$ acting on $\ell^2(\mathbb{Z}^d)$  where $\Delta_{[\infty]^d}$ is the discrete Laplacian on $\ZZ^d$, $\lambda>0$ the coupling constant, and $V_{\omega}$ a random potential $V_{\omega} = \{V_{\omega}(n): n  \in \mathbb{Z}^d\}$ consisting of i.i.d. variables with common probability distribution with, for our purposes, bounded density $\mu$ on $(0,\infty)$. Here, $\omega$ is an element of the product probability space $\Omega = (\text{supp}(\mu) )^{\mathbb{Z}^d} $ endowed with the $\sigma$- algebra generated by the cylinder sets and the product measure $ \mu^{\mathbb{Z}^d}$ consisting of the common probability distribution with compact support. The random potential $V_{\omega}: \mathbb{Z}^d \to \mathbb{R} $ is defined as projections $ \Omega \ni \omega \mapsto V_{\omega}(n)=\omega_n $ for $n \in \mathbb{Z}^d$.  
\\
We also consider $H_{\omega,\lambda}^{[N]^d}$ the restriction to finite domains of size $[N]^d$, of the operator $H^{[\infty]^d}_{\omega,\lambda}$, with Neumann boundary conditions.\\

\begin{figure}
\includegraphics[width=10cm]{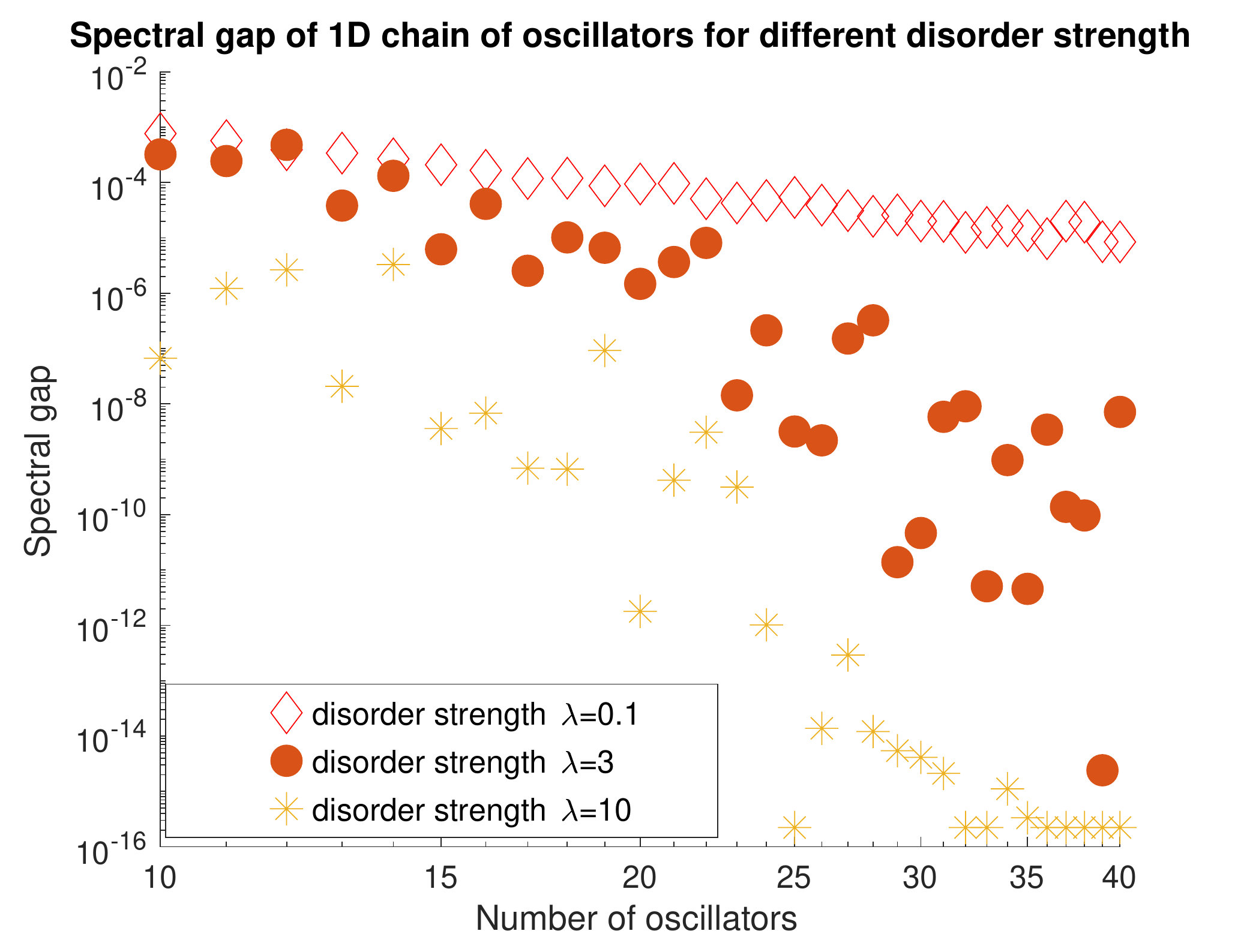}
\caption{Log-log plot of the spectral gap for the one-dimensional chain of oscillators for different disorder strengths when $\eta_i \sim 1+ \lambda U[0,1]$.}
\end{figure}

So the spectral gap of the N-dimensional disordered chain of $[N]^d$ oscillators coupled at two heat baths at different temperatures, as described above, is given by one of the points where the Wigner $\widetilde R_{I}$-matrix is not invertible. Since this lower-dimensional matrix is defined in terms of the eigenvalues and eigenvectors of the block matrix $B_{[N]^d}$, see Lemma \ref{lemm:SGI}, we are interested in the spectrum of $B_{[N]^d}$ which can be identified with $H_{\omega,\lambda}^{[N]^d}$.  More specifically, 
%So, the one-dimensional disordered chain of $N$ oscillators as described above, is \emph{essentially} described by the operator $H_{\omega}^N$. 
the deterministic discrete Laplacian, restricted to a domain of size $[N]^d$, describes the deterministic two-body interactions, while the random potential represents a disordered in the pinning potential. In the $N \to \infty$-limit, this model reduces to the Anderson model. \\

In one dimension, the Anderson model has a.s\@ dense pure point spectrum with exponentially localized eigenstates \cite{FS83, DK89}. In higher dimensions, $d \ge 2$ this is only known to be true for sufficiently large disorder or low energies and was already shown in \cite{FS83}. From the case of a single impurity we know already that exponentially localized eigenstates should lead to an exponentially fast closing of the spectral gap. However, we have to deal with three additional obstructions in the disordered case: 
\begin{itemize}
\item The eigenvalues of the Anderson model are not uniformly (in $N$) bounded away from each other. 
\item The eigenfunctions of the Anderson model do not obey a rich symmetry as before and can (in general) not be chosen to be even or odd. 
\item We are studying finite approximations $B_{[N]^d}$ rather than the Anderson model $B_{[\infty]^d}=H^{[\infty]^d}_{\omega,\lambda}$ itself.
\end{itemize}
The next Lemma shows that in general eigenvalues will not get any closer than a distance $N^{-2d-2}.$
\begin{lemm} \label{Minami} Let $A_N(s([N]^d))$ be the event that for the $N^d$-size Anderson model $H^{[\pm N]^d}_{\omega}$, there exists an interval of size $s([N]^d)$ that contains (at least) two eigenvalues.
For the choice $s([N]^d) = N^{-2d-2}$ we have $\mathbb P(A_N(s([N]^d)))=0$ for all but finitely many $N.$
 \end{lemm}
 
  \begin{proof}
The spectrum of $H_{\omega,\lambda}^{[N]^d}$ is contained in an interval of order one. Thus, we can cover $\Spec(H_{\omega,\lambda}^{[N]^d})$ by $\mathcal O(1/s([N]^d))$ many intervals $(I_{n}^{[N]^d})_{n \in [ \mathcal O(1/s([N]^d))]}$ of size $2s([N]^d)$ such that the overlap of each interval $I_n^{[N]^d}$ with its nearest neighbors is another interval of size $s([N]^d).$ This construction implies that if there exists an interval of size $s([N]^d)$ that contains two eigenvalues, these two eigenvalues are also contained in one of the $I_n^{[N]^d}.$ \\ 
We will now use Minami's estimate which bounds from above the probability of two eigenvalues of the finite volume operator being close, see \cite[(7), App.\@ 2]{KM06}. More specifically that is $$ \mathbb P \big( \text{number of eigenvalues in a set}\ J\ \text{is } \geq  2 \big) \leq \pi^2 \|\rho\|_{\infty}^2  N^{2d}|J|^2, $$ we write 
\begin{equation}
\begin{split}
\label{eq:prob}
\mathbb P\left(A_N(s([N]^d)) \right) 
&\le \sum_{n \in [ \mathcal O(1/s([N]^d))]} \mathbb P\left( \vert I_n^N \cap \Spec(H^{[N]^d}_{\omega}) \vert \ge 2\right) \\
&\le  \sum_{n \in [ \mathcal O(1/s([N]^d))]}\pi^2 \| \rho \|_{\infty}^2  N^{2d} 4 s([N]^d)^2 \\ &= \mathcal O(N^{2d} s([N]^d))< \infty
\end{split}
\end{equation}
We choose now $s([N]^d) = N^{-2d-2}$, such that by the Borel-Cantelli lemma $A_N(s([N]^d))$ happens at most finitely many times a.s.\@ and otherwise eigenvalues of $H^{[N]^d}_{\lambda,\omega}$ are a.s.\@ at least $N^{-2d-2}$ apart.
\end{proof}

With this Lemma at hand, we can now give the proof of the exponential decay of the spectral gap.

\begin{prop}
\label{prop:Anders}
Consider the chain of oscillators with equal masses, unit interaction strength, and non-zero friction at at least one end of the chain. In addition, we assume that there is always at least one particle experiencing friction at the boundary and that the friction of particles is uniformly bounded in $N$. Let the pinning constants be iid $\eta \sim \rho \in C_c(0,\infty).$ Then the spectral gap of the chain of oscillators decays, for almost every realization of the disorder in the pinning potential, exponentially fast\footnote{The decay of the spectral gap will in general depend on the disorder but is a.s. exponentially fast.}.
\end{prop}
\begin{proof}
For almost every realization of disorder we can find by general results on the Anderson model  \cite{FS83, DK89} an eigenfunction $ \varphi$ of the operator $B_{[\infty]^d}$, corresponding to an eigenvalue $\lambda_{\infty}$ such that 
$$ \sup_{i; \in \Vert i \Vert_{\infty}=N} \vert \varphi( i ) \vert = \mathcal O(e^{-DN})\ \text{and}\ \sum_{m \notin [\pm N]^d} \vert \varphi (m) \vert^2 = \mathcal{O}(e^{-DN}).$$

\medskip

By using Lemma \ref{lemm:jess} with $m_0=0$ and Lemma \ref{Minami}  it follows that for all but finitely many $N$ the distance between any two eigenvalues is at least $\alpha_N:=N^{-{2(d+1)}}$ and we find an eigenvector $\widetilde{\psi}_{[N]^d}$ to $B_{[N]^d}$ with eigenvalue $\widetilde{\lambda}_N$ that approximates $\varphi$ with eigenvalue $\lambda_{\infty}$. Thus, for all but finitely many $N$ 
\[ \Vert \widetilde{\psi}_{[N]^d}- \widetilde \varphi \vert _{[N]^d}  \Vert = \mathcal O(e^{-DN}N^{2(d+1)}). \]

\medskip

As before, we shall restrict us again to the case $d= 2$ for simplicity and study solutions of the equivalent problem equation \begin{align} \label{eq: reduced prob2} \operatorname{det}(F(\lambda) + G(\lambda) )=0\end{align} in terms of the vectors $V_{j}^{\pm}=\frac{1}{\sqrt{2}}(v_{j}, \pm i v_{j})^T$ and $\mu_{j}= \lambda_{j}- \lambda_{1}$, where $v_j$ are the eigenvectors of the Schr\"odinger operator $B_{[N]^d}$ with eigenvalue $\lambda_j$ and $\lambda_1$ being the eigenvalue associated with $\widetilde{\psi}_{[N]^d}$.  The matrices $F(\lambda), G(\lambda)$ are then defined as follows 
\begin{equation} 
\label{eq:F2}
F(\lambda):= -i \lambda - \lambda \sum_{j=2}^{N^2} \sum_{\pm}\sum_{\substack{i_1,i_2\in I}} \frac{ \sqrt{\gamma_{i_1}\gamma_{i_2}} }{\mu_{j}} \langle V_{j}^{\pm}, e_{i_1}^{2N^2}\rangle \langle e_{i_2}^{2N^2} ,V_{j}^{\pm}\rangle e_{i_1}^{\vert I \vert} \otimes e_{i_2}^{\vert I \vert}
\end{equation} and 
\begin{equation}
\begin{split}
G(\lambda)&:=\underbrace{\sum_{\pm}\sum_{\substack{i_1,i_2 \in I}}  \sqrt{\gamma_{i_1}\gamma_{i_2}}  \langle V_{1}^{\pm}, e_{i_1}^{2N^2}\rangle \langle e_{i_2}^{2N^2} ,V_{1}^{\pm}\rangle e_{i_1}^{\vert I \vert} \otimes e_{i_2}^{\vert I \vert}}_{=:(\text{I})} \\
 &\underbrace{ \ - \lambda^2 \sum_{j=2}^{N^2} \sum_{\pm}\sum_{\substack{i_1,i_2 \in I}} \frac{\sqrt{\gamma_{i_1}\gamma_{i_2}} }{\mu_{j}^2} \langle V_{j}^{\pm}, e_{i_1}^{2N^2}\rangle \langle e_{i_2}^{2N^2} ,V_{j}^{\pm}\rangle e_{i_1}^{\vert I \vert} \otimes e_{i_2}^{\vert I \vert}}_{=:(\text{II})}
\\
  & \underbrace{- \lambda^3 \sum_{j=2}^{N^2} \sum_{\pm}\sum_{\substack{i_1,i_2 \in I}} \frac{\sqrt{\gamma_{i_1}\gamma_{i_2}}}{\mu_{j}^2(\mu_{j}-\lambda)} \langle V_{j}^{\pm}, e_{i_1}^{2N^2}\rangle \langle e_{i_2}^{2N^2} ,V_{j}^{\pm}\rangle e_{i_1}^{\vert I \vert} \otimes e_{i_2}^{\vert I \vert}}_{=:(\text{III})}
\end{split}
\end{equation}
so that a solution to \eqref{eq: reduced prob2} corresponds to the desired eigenvalue. Before we fix a ball $K=B(0,r_N)$, we again want to find an upper bound for the $\| \Theta_N\|$, where 
\begin{equation}
\label{eq:ThetaN2}
\Theta_N:= \sum_{\pm}\sum_{\substack{i_1,i_2 \in I}}  \sqrt{\gamma_{i_1}\gamma_{i_2}}    \langle V_{1}^{\pm}, e_{i_1}^{2N^2}\rangle \langle e_{i_2}^{2N^2} ,V_{1}^{\pm}\rangle e_{i_1}^{2N} \otimes e_{i_2}^{2N} 
\end{equation}
 i.e. is the first term, $(\text{I})$, of $G$. 
Using exponential decay of the eigenstate $V_1^{\pm}$ it follows that for $c>0$ we have 
\[ \Vert \Theta_N \Vert =\mathcal O(N e^{-cN}).\]

We now fix a ball $K=B(0,r_N)$ and choose the radius $r_N:= \mathcal O(Ne^{-cN})$.  
Therefore it suffices to find a root of \eqref{eq: reduced prob} inside the ball $K$ and conclude the existence of an eigenvalue by Rouch\'e's theorem. We easily see that for all $v \neq 0$, and  $\lambda \in \partial K$,   $\|F(\lambda) v \| \ge |\lambda| \|v \| = r_N \|v\|$ since the second term of the right hand side of \eqref{eq:F2} is symmetric. On the other hand, by Lemma \ref{Minami} we can estimate $\mu_j \ge cN^{-2(d+1)}$ for all but finitely many $N$ and some $c>0$. Using this lower bound, we have for $\lambda \in \partial K$
\begin{align*}
\Vert (\text{II}) \Vert = \mathcal O(N^{2(d+2)}e^{-2cN}) \text{ and }\Vert (\text{III}) \Vert = \mathcal O(N^{2(d+2)}e^{-3cN}).
\end{align*}

Thus, we have $\| F(\lambda)v \| > \| G(\lambda) v\|$ on $\partial K$, for all $v \neq 0$ for almost all sufficiently large $N$. Since $F(\lambda)$ is not invertible exactly at $0$ inside $K$, we have from Lemma \ref{lemm:Rouch} that there is one point inside $K$ so that $F(\lambda)+ G(\lambda) $ is not invertible or in other words there is one root of $R_I(\lambda)-iu=0$ with $\lambda \lesssim N e^{-cN}$. 
\end{proof}

\begin{appendix}
\section{Matrix-valued Rouch\'e's theorem}

\begin{lemm}[Matrix-valued Rouch\'e's theorem]
\label{lemm:Rouch}
Let $A,B: K \rightarrow \mathbb C^{n \times n}$ be two holomorphic functions inside some region $K$ with $\Vert B(z)v \Vert <\Vert A(z)v \Vert$ for all $v \neq 0$ and $z \in \partial K.$ Then, both $A$ and $A+B$ are invertible at an equal number of points inside $K.$
\end{lemm}
\begin{proof}
By the argument principle the number of singular points of $A(z)+tB(z)$ in $K$ with $t \in [0,1]$ is given by  
\[N(t):=\frac{1}{2\pi i}  \int_{\partial K }  \partial_z \log( \operatorname{det}(A(z)+tB(z)) ) \ dz\]
and independent of $t$ by continuity of $t \mapsto N(t).$
\end{proof}
\end{appendix}
\smallsection{Acknowledgements} 
This work was supported by the EPSRC grant EP/L016516/1 for the University of Cambridge CDT, the CCA. The authors are grateful to Cl\'ement Mouhot for several discussions. We also thank Pierre Monmarch\'e for sharing some of his insights with us. 
%%%%%%%%%%%%%%%%%%%%%%%%%%%%%%%%%%%%%%%%%%%%%%%%%%%%%%%%%%%%%%%%%%%%%%%%%%%%%%%%
%                                 BIBLIOGRAPHY                                 %
%%%%%%%%%%%%%%%%%%%%%%%%%%%%%%%%%%%%%%%%%%%%%%%%%%%%%%%%%%%%%%%%%%%%%%%%%%%%%%%%

\bibliographystyle{alpha}
\bibliography{bibliography}

\begin{thebibliography}{CEHRB18}

\bibitem[AAS15]{AAS15}
F.~Achleitner, A.~Arnold, and D.~St\"{u}rzer.
\newblock Large-time behavior in non-symmetric {F}okker-{P}lanck equations.
\newblock {\em Riv. Math. Univ. Parma (N.S.)}, 6(1):1--68, 2015.

\bibitem[AE]{AE14}
A.Arnold and J.~Erb.
\newblock Sharp entropy decay for hypocoercive and non-symmetric fokker-planck
  equations with linear drift.
\newblock https://arxiv.org/abs/1409.5425.

\bibitem[AH11]{AH10}
O.~Ajanki and F.~Huveneers.
\newblock Rigorous scaling law for the heat current in disordered harmonic
  chain.
\newblock {\em Comm. Math. Phys.}, 301(3):841--883, 2011.

\bibitem[And58]{AN58}
P.~W. Anderson.
\newblock Absence of diffusion in certain random lattices.
\newblock {\em Phys. Rev.}, 109:1492--1505, Mar 1958.

\bibitem[BHO19]{BHO19}
C.~Bernardin, F.~Huveneers, and S.~Olla.
\newblock Hydrodynamic limit for a disordered harmonic chain.
\newblock {\em Comm. Math. Phys.}, 365(1):215--237, 2019.

\bibitem[BLRB00]{BLR00}
F.~Bonetto, J.~L. Lebowitz, and L.~Rey-Bellet.
\newblock Fourier's law: a challenge to theorists.
\newblock In {\em Mathematical physics 2000}, pages 128--150. Imp. Coll. Press,
  London, 2000.

\bibitem[Car07]{Car07}
P.~Carmona.
\newblock Existence and uniqueness of an invariant measure for a chain of
  oscillators in contact with two heat baths.
\newblock {\em Stochastic Process. Appl.}, 117(8):1076--1092, 2007.

\bibitem[CEHRB18]{CEHRB18}
N.~Cuneo, J.-P. Eckmann, M.~Hairer, and L.~Rey-Bellet.
\newblock Non-equilibrium steady states for networks of oscillators.
\newblock {\em Electron. J. Probab.}, 23:28 pp., 2018.

\bibitem[CEP15]{CEP15}
N.~Cuneo, J.-P. Eckmann, and C.~Poquet.
\newblock Non-equilibrium steady state and subgeometric ergodicity for a chain
  of three coupled rotors.
\newblock {\em Nonlinearity}, 28(7):2397--2421, 2015.

\bibitem[CL]{CL71}
A.~Casher and J.~L. Lebowitz.
\newblock Heat flow in regular and disordered harmonic chains.
\newblock Journal of Mathematical Physics 12, 1701 (1971).

\bibitem[CP17]{CP17}
N.~Cuneo and C.~Poquet.
\newblock On the relaxation rate of short chains of rotors interacting with
  {L}angevin thermostats.
\newblock {\em Electron. Commun. Probab.}, 22:Paper No. 35, 8, 2017.

\bibitem[Dha01]{Dhar01}
A.~Dhar.
\newblock Heat conduction in the disordered harmonic chain revisited.
\newblock {\em Physical review letters}, 86:5882--5, 07 2001.

\bibitem[Dha08]{Dhar08}
A.~Dhar.
\newblock Heat transport in low-dimensional systems.
\newblock {\em Advances in Physics}, 57, 08 2008.

\bibitem[DKS83]{DKS}
F.~Delyon, H.~Kunz, and B.~Souillard.
\newblock One-dimensional wave equations in disordered media.
\newblock {\em J. Phys. A}, 16(1):25--42, 1983.

\bibitem[DSS87]{DSS}
F.~Delyon, B.~Simon, and B.~Souillard.
\newblock Localization for off-diagonal disorder and for continuous
  {S}chr\"{o}dinger operators.
\newblock {\em Comm. Math. Phys.}, 109(1):157--165, 1987.

\bibitem[EH00]{EH00}
J.-P. Eckmann and M.~Hairer.
\newblock Non-equilibrium statistical mechanics of strongly anharmonic chains
  of oscillators.
\newblock {\em Comm. Math. Phys.}, 212(1):105--164, 2000.

\bibitem[EPRB99a]{EPR99a}
J.-P. Eckmann, C.-A. Pillet, and L.~Rey-Bellet.
\newblock Entropy production in nonlinear, thermally driven {H}amiltonian
  systems.
\newblock {\em J. Statist. Phys.}, 95(1-2):305--331, 1999.

\bibitem[EPRB99b]{EPR99b}
J.-P. Eckmann, C.-A. Pillet, and L.~Rey-Bellet.
\newblock Non-equilibrium statistical mechanics of anharmonic chains coupled to
  two heat baths at different temperatures.
\newblock {\em Comm. Math. Phys.}, 201(3):657--697, 1999.

\bibitem[FB19]{BF19}
P.~Flandrin and C.~Bernardin, editors.
\newblock {\em Fourier and the {S}cience of {T}oday / {F}ourier et la {S}cience
  d’aujourd’hui}, volume 20, Issue 5.
\newblock Comptes Rendus Physique, 2019.

\bibitem[FKS97]{FYODOROV199746}
Yan~V Fyodorov, Boris~A Khoruzhenko, and Hans-Jürgen Sommers.
\newblock Almost-hermitian random matrices: eigenvalue density in the complex
  plane.
\newblock {\em Physics Letters A}, 226(1):46 -- 52, 1997.

\bibitem[FS83]{FS83}
J.~Fr\"{o}hlich and T.~Spencer.
\newblock Absence of diffusion in the {A}nderson tight binding model for large
  disorder or low energy.
\newblock {\em Comm. Math. Phys.}, 88(2):151--184, 1983.

\bibitem[Hai09]{Hair09}
M.~Hairer.
\newblock How hot can a heat bath get?
\newblock {\em Comm. Math. Phys.}, 292(1):131--177, 2009.

\bibitem[Hel71]{Hellemann}
R.~H.~G. Hellemann.
\newblock Heat transport in harmonic and anharmonic lattices.
\newblock thesis, Yeshiva University, 1971.

\bibitem[HM09]{HM09}
M.~Hairer and J.~Mattingly.
\newblock Slow energy dissipation in anharmonic oscillator chains.
\newblock {\em Comm. Pure Appl. Math.}, 62(8):999--1032, 2009.

\bibitem[Hö67]{H69}
L.~Hörmander.
\newblock Hypoelliptic second order differential equations.
\newblock {\em Acta Math.}, 119:147--171, 1967.

\bibitem[JPS17]{JPS17}
V.~Jak\v{s}i\'{c}, C.-A. Pillet, and A.~Shirikyan.
\newblock Entropic fluctuations in thermally driven harmonic networks.
\newblock {\em J. Stat. Phys.}, 166(3-4):926--1015, 2017.

\bibitem[JRG71]{RG71}
R.~J.~Rubin and W.~Greer.
\newblock Abnormal lattice thermal conductivity of a one‐dimensional,
  harmonic, isotopically disordered crystal.
\newblock {\em Journal of Mathematical Physics}, 12:1686--1701, 08 1971.

\bibitem[KM06]{KM06}
A.~Klein and S.~Molchanov.
\newblock Simplicity of eigenvalues in the {A}nderson model.
\newblock {\em J. Stat. Phys.}, 122(1):95--99, 2006.

\bibitem[Lep16]{Lep16}
S.~Lepri, editor.
\newblock {\em Thermal {T}ransport in {L}ow {D}imensions}, volume 921 of {\em
  Lecture Notes in Physics}.
\newblock Springer, [Cham], 2016.
\newblock From statistical physics to nanoscale heat transfer.

\bibitem[Men20]{Me20}
Angeliki Menegaki.
\newblock Quantitative {R}ates of {C}onvergence to {N}on-equilibrium {S}teady
  {S}tate for a {W}eakly {A}nharmonic {C}hain of {O}scillators.
\newblock {\em J. Stat. Phys.}, 181(1):53--94, 2020.

\bibitem[Mon19]{Mon15}
P.~Monmarch\'{e}.
\newblock Generalized {$\Gamma$} calculus and application to interacting
  particles on a graph.
\newblock {\em Potential Anal.}, 50(3):439--466, 2019.

\bibitem[MPP02]{MPP02}
G.~Metafune, D.~Pallara, and E.~Priola.
\newblock Spectrum of {O}rnstein-{U}hlenbeck operators in {$L^p$} spaces with
  respect to invariant measures.
\newblock {\em J. Funct. Anal.}, 196(1):40--60, 2002.

\bibitem[NR]{NR19}
V.~Nersesyan and R.~Raquépas.
\newblock Exponential mixing under controllability conditions for {S}{D}{E}s
  driven by a degenerate {P}oisson noise.
\newblock https://arxiv.org/abs/1903.08089.

\bibitem[OL74]{CL74}
A.~J. O'Connor and J.~L. Lebowitz.
\newblock Heat conduction and sound transmission in isotopically disordered
  harmonic crystals.
\newblock {\em J. Mathematical Phys.}, 15:692--703, 1974.

\bibitem[Raq19]{RAQ18}
R.~Raqu\'{e}pas.
\newblock A note on {H}arris' ergodic theorem, controllability and
  perturbations of harmonic networks.
\newblock {\em Ann. Henri Poincar\'{e}}, 20(2):605--629, 2019.

\bibitem[RBT02]{RBT02}
L.~Rey-Bellet and L.~Thomas.
\newblock Exponential convergence to non-equilibrium stationary states in
  classical statistical mechanics.
\newblock {\em Comm. Math. Phys.}, 225(2):305--329, 2002.

\bibitem[RLL67]{RLL67}
Z.~Rieder, J.~L. Lebowitz, and E.~Lieb.
\newblock Properties of a harmonic crystal in a stationary nonequilibrium
  state.
\newblock {\em Journal of Mathematical Physics}, 8(5):1073--1078, 1967.

\bibitem[SZ89]{SZ}
V.V. Sokolov and V.G. Zelevinsky.
\newblock Dynamics and statistics of unstable quantum states.
\newblock {\em Nuclear Physics A}, 504(3):562 -- 588, 1989.

\bibitem[Tes00]{T}
G.~Teschl.
\newblock {\em Jacobi operators and completely integrable nonlinear lattices},
  volume~72 of {\em Mathematical Surveys and Monographs}.
\newblock American Mathematical Society, Providence, RI, 2000.

\bibitem[vDK89]{DK89}
H.~von Dreifus and A.~Klein.
\newblock A new proof of localization in the {A}nderson tight binding model.
\newblock {\em Comm. Math. Phys.}, 124(2):285--299, 1989.

\bibitem[Ver79]{Ver79}
T.~Verheggen.
\newblock Transmission coefficient and heat conduction of a harmonic chain with
  random masses: asymptotic estimates on products of random matrices.
\newblock {\em Comm. Math. Phys.}, 68(1):69--82, 1979.

\bibitem[Vil09]{Villani09}
C\'{e}dric Villani.
\newblock Hypocoercivity.
\newblock {\em Mem. Amer. Math. Soc.}, 202(950):iv+141, 2009.

\end{thebibliography}

\end{document}